\documentclass[a4paper]{article}

\usepackage{babel}
\usepackage{bm}
\usepackage[utf8]{inputenc}
\usepackage[T1]{fontenc}
\usepackage{subfig}
\usepackage{fancyhdr}
\usepackage{booktabs}
\usepackage{listings}
\usepackage{graphicx}
\usepackage{longtable}
\usepackage{pdflscape}
\usepackage{tcolorbox}
\usepackage{multirow}
\usepackage{cite}
\usepackage[table]{xcolor}  

\usepackage[a4paper,top=3cm,bottom=2cm,left=2cm,right=2cm,marginparwidth=1.75cm]{geometry}

\usepackage{amsmath}
\usepackage{amssymb}
\usepackage{indentfirst}
\usepackage{amsfonts}
\usepackage{graphicx}
\usepackage{float} 
\usepackage[ruled,vlined]{algorithm2e}
\usepackage[colorinlistoftodos]{todonotes}
\usepackage{caption}
\usepackage{amsthm}
\usepackage{sectsty}
\usepackage{apacite}
\usepackage{float}
\usepackage{titling} 
\usepackage{blindtext}
\usepackage[numbers,sort&compress]{natbib}
\usepackage[colorlinks=true, allcolors=blue]{hyperref}
\usepackage[colorinlistoftodos]{todonotes}
\usepackage{color}
\usepackage{xcolor}
\usepackage{setspace}
\usepackage{enumerate}
\usepackage{fontspec}
\usepackage{tikz}
\usepackage{mathrsfs}
\usepackage{esint}
\usepackage{bm}
\usepackage{physics}

\usepackage{titlesec}  
\usepackage{titletoc}
\usepackage{etoolbox}

\usetikzlibrary{shapes,arrows,positioning,backgrounds,fit,trees,calc,shapes.geometric}

\definecolor{backcolour}{rgb}{0.95,0.95,0.92}
\definecolor{codegreen}{rgb}{0,0.6,0}
\definecolor{codegray}{rgb}{0.5,0.5,0.5}
\definecolor{codepurple}{rgb}{0.58,0,0.82}
\definecolor{magenta}{rgb}{0.6,0,0.6}
\lstdefinestyle{mystyle}{
    backgroundcolor=\color{backcolour},   
    commentstyle=\color{codegreen},
    keywordstyle=\color{magenta},
    numberstyle=\tiny\color{codegray},
    stringstyle=\color{codepurple},
    basicstyle=\ttfamily\footnotesize,
    breakatwhitespace=false,         
    breaklines=true,                 
    captionpos=b,                    
    keepspaces=true,                 
    numbers=left,                    
    numbersep=5pt,                  
    showspaces=false,                
    showstringspaces=false,
    showtabs=false,                  
    tabsize=2
}

\newtheorem{theorem}{\indent Theorem}[section]

\newtheorem{proposition}[theorem]{\indent Proposition}

\theoremstyle{definition}
\newtheorem{definition}{\indent Definition}[section]
\newtheorem{example}{\indent Example}[section]
\newtheorem{remark}{\indent Remark}[section]

\newcounter{chapter}
\renewcommand{\thechapter}{\arabic{chapter}}

\titleformat{\chapter}[block]
  {\normalfont\LARGE\bfseries\raggedright}
  {Chapter \thechapter\quad}
  {0pt}
  {}
  
\titlespacing*{\chapter}{0pt}{20pt}{10pt}

\titlecontents{chapter}
  [0pt]
  {\addvspace{8pt}\bfseries\color{blue}}
  {\contentslabel{0pt}\color{blue}\thecontentslabel\quad\color{blue}}
  {}
  {\hfill\color{black}\contentspage} 
  
\titlecontents{section}
  [1.5em]
  {\addvspace{4pt}\color{blue}}
  {\contentslabel{1.5em}\color{blue}}
  {}
  {\hfill\color{black}\contentspage} 
  
\begin{document}
	
\title{\LARGE \bfseries Optical Caustics as Lagrangian Singularities: Classification and Geometric Structure}

\author{
    \large Rongqi Shang\textsuperscript{1}, Donglin Ma\textsuperscript{1,2}\thanks{Corresponding author: madonglin@hust.edu.cn} \\[1em]
    \small \textsuperscript{1}School of Mathematics and Statistics, Huazhong University of Science and Technology, Wuhan 430074, China \\
    \small \textsuperscript{2}School of Optical and Electronic Information, Huazhong University of Science and Technology, Wuhan 430074, China
}

\date{}

\maketitle

\begin{abstract}
This paper develops a rigorous mathematical framework for light propagation by constructing the optical phase space with its symplectic structure and the extended phase space with its contact structure. We prove that light rays in three-dimensional Euclidean space correspond to Reeb orbits in a five-dimensional contact manifold, which are then projected onto a four-dimensional symplectic manifold via symplectic reduction. Leveraging the advantages of phase space, we provide a rigorous definition of caustic surfaces as singularities of the Lagrangian submanifold projection and derive explicit expressions for caustic surfaces in convex lens systems. Furthermore, based on singularity theory, we present a complete classification of stable caustic surfaces and establish a correspondence with classical Seidel aberration theory. Building upon this theory, we propose a method of \emph{topological optical correction} that overcomes the limitations of traditional optimization algorithms in dealing with complex caustic structures. This work provides a new mathematical paradigm for the design and correction of high-precision optical systems.

\end{abstract}

\section{Introduction}
\label{sec:introduction}

\subsection{Background}

Geometric optics models light propagation by tracing individual rays. Hamiltonian mechanics offers a powerful geometric perspective for this model. In this framework, the set of all possible rays---each determined by its position and direction---forms a region in the \textbf{optical phase space}, which naturally carries a \textbf{symplectic structure}.

When the optical medium is inhomogeneous, or when one wishes to treat the evolution parameter (e.g., the optical axis coordinate $z$) on an equal footing, it is often advantageous to work in the \textbf{extended phase space}, which possesses a \textbf{contact structure}. Caustic surfaces, as envelopes of families of rays, naturally appear as singularities of the projection from these spaces to physical space. Traditional geometric optics relies on solving ray equations, which is effective for analyzing conventional imaging systems. However, when analyzing behavior near foci, regions with strong aberrations, and the distribution of singularities in light field intensity, the ray-based description encounters fundamental difficulties---the light intensity diverges on caustic surfaces. This divergence indicates a breakdown of the mathematical model at singular points.

To overcome this limitation, since the mid-20th century, mathematicians such as V. I. Arnold \cite{arnold1990contact,arnold2000symplectic,arnold1985singularities} and R. Thom \cite{Gol} have introduced symplectic geometry and catastrophe theory into optics. Symplectic geometric optics no longer treats rays as isolated trajectories but as points in a four-dimensional phase space or as curves in a contact manifold. Caustic surfaces are redefined as critical value sets of the projection of Lagrangian submanifolds in phase space onto configuration space.

This shift in perspective is profound: it reveals that optical aberrations are not merely ``errors'' but intrinsic breakdowns of the system's structural stability, whose forms are constrained by deep topological invariants.

\subsection{Motivation and Structure}

This report aims to establish a rigorous mathematical framework that integrates modern mathematical physics---particularly symplectic and contact geometry---into the analysis of optical aberrations, and to precisely correlate classical concepts in optical engineering (Seidel aberrations and Zernike polynomials) with normal forms in catastrophe theory. This correspondence is not merely morphological but an isomorphism of algebraic structures.

By deriving the mathematical generation mechanisms of $A_k$ and $D_k$ series catastrophes, we gain deep insight into the topological origins of aberrations such as spherical aberration, coma, and astigmatism in the formation of caustic surfaces. Furthermore, to address the challenge of correcting high-order aberrations in modern high-precision optical systems, this report proposes a general optical correction method, termed \emph{Topological Optical Correction} (TOC). This method uses Zernike coefficients as control parameters to plan optimal correction paths within the geometric structure of bifurcation sets, thereby overcoming the non-convexity and tendency of traditional least-squares methods to get stuck in local minima when dealing with complex caustic structures.

The report is structured as follows. Section~\ref{sec:symplectic} constructs the Hamiltonian formulation of symplectic optics. Section~\ref{sec:contact} introduces the contact geometry of the extended phase space. Section~\ref{sec:3d-phase-space} details the reduction from a six-dimensional symplectic space to a four-dimensional optical phase space. Section~\ref{sec:wavefronts} discusses the representation of wavefronts in geometric optics. Section~\ref{sec:caustics} provides a rigorous mathematical theory of caustic surfaces. Section~\ref{sec:lens-caustic} derives explicit expressions for caustic surfaces in convex lens systems. Section~\ref{sec:classification} presents a complete classification of caustic surfaces based on singularity theory. Section~\ref{sec:correspondence} establishes the correspondence between caustic surfaces and Seidel aberrations. Section~\ref{sec:correction} proposes the TOC method for optical correction. Finally, Section~\ref{sec:conclusion} summarizes the main contributions and suggests future research directions.

\section{Symplectic Geometry of Optical Phase Space}
\label{sec:symplectic}

\begin{definition}[Configuration Space]
Let $M$ be a two-dimensional smooth manifold. In planar geometry, $M$ can be the plane $\mathbb{R}^2$ (a screen), or for rotationally symmetric systems, a region on the sphere $S^2$ (direction space). The \textbf{optical phase space} is defined as the cotangent bundle $\mathscr{P} = T^*M$ of this manifold.
\end{definition}

Points in $\mathscr{P}$ are denoted $(\bm{q}, \bm{p})$, where $\bm{q} \in M$ denotes position coordinates and $\bm{p} \in T^*_q M$ is the optical momentum. In a local chart $(q^1, q^2)$ of $M$, the momentum $\bm{p}$ has components $(p_1, p_2)$. For a ray passing through point $\bm{q}$ in a medium with local refractive index $n(\bm{q})$ and with direction vector $\dot{\bm{q}}$, the momentum is given by $\bm{p} = n(\bm{q}) \, \dot{\bm{q}} / \|\dot{\bm{q}}\|$ (under normalized parameterization), satisfying the constraint $\|\bm{p}\|^2 = n(\bm{q})^2$.

\begin{definition}[Canonical Symplectic Form]
The manifold $\mathscr{P} = T^*M$ carries a \textbf{canonical symplectic form} $\omega$, which in local coordinates $(q^1, q^2, p_1, p_2)$ is expressed as:
\[
\omega = dp_1 \wedge dq^1 + dp_2 \wedge dq^2.
\]
This 2-form is closed ($d\omega = 0$) and non-degenerate, making $(\mathscr{P}, \omega)$ a \textbf{symplectic manifold}.
\end{definition}

\begin{definition}[Hamiltonian Function and Hamiltonian Flow]
Let $H : \mathscr{P} \to \mathbb{R}$ be a smooth function, called the \textbf{optical Hamiltonian}. For homogeneous media, a common choice is $H(\bm{q}, \bm{p}) = \|\bm{p}\|^2 - n^2$, which constrains rays to the level set $H=0$. The Hamiltonian vector field $X_H$ associated with $H$ is implicitly defined by:
\[
\iota_{X_H} \omega = -dH,
\]
where $\iota$ denotes contraction. In local coordinates, this yields Hamilton's equations:
\begin{align}
\frac{dq^i}{dz} &= \frac{\partial H}{\partial p_i}, \\
\frac{dp_i}{dz} &= -\frac{\partial H}{\partial q^i},
\end{align}
where $z$ is a parameter along the ray (often the coordinate along the optical axis). The flow generated by $X_H$ consists of \textbf{symplectomorphisms} (canonical transformations) that preserve the symplectic form, i.e., $\phi_z^* \omega = \omega$.
\end{definition}

\section{Contact Geometry of Extended Phase Space}
\label{sec:contact}

\begin{definition}[Extended Optical Phase Space]
The \textbf{extended phase space} is defined as $\mathscr{E} = T^*M \times \mathbb{R}$, where the additional $\mathbb{R}$-coordinate, denoted $t$ or $z$, is the evolution parameter. Points in $\mathscr{E}$ are $(\bm{q}, \bm{p}, t)$. This is a $(2\dim M + 1) = 5$-dimensional manifold.
\end{definition}

\begin{definition}[Canonical Contact Form]
The manifold $\mathscr{E}$ carries a \textbf{canonical contact form}. In local coordinates $(q^1, q^2, p_1, p_2, t)$, this 1-form is defined as:
\begin{equation}
    \alpha = p_1 dq^1 + p_2 dq^2 - H(\bm{q}, \bm{p}, t) \, dt.
\end{equation}
Here,
\begin{itemize}
    \item $p_1 dq^1 + p_2 dq^2$ is the work form, computing the work done by the momentum along the displacement direction.
    \item $- H(\bm{q}, \bm{p}, t) \, dt$ is the energy form, where $H dt$ represents the change in system energy over time $dt$.
\end{itemize}

In classical mechanics, the familiar action is defined as:
\begin{equation}
    S = \int_{t_1}^{t_2} L(\bm{q}, \dot{\bm{q}}, t) \, dt = \int_{t_1}^{t_2} [\bm{p}\cdot\dot{\bm{q}} - H(\bm{q},\bm{p},t)] \, dt= \int [\bm{p} \cdot d\bm{q} - H \, dt]=\int_\gamma \alpha,
\end{equation}
where the second equality uses the Legendre transform $L = \bm{p}\cdot\dot{\bm{q}} - H$.

If $H$ is non-singular and $\alpha \wedge (d\alpha)^{\wedge n} \neq 0$ is a volume form, then the 5-dimensional space $(\mathscr{E}, \alpha)$ is a \textbf{contact manifold}.
\end{definition}

The kernel of $\alpha$ defines the \textbf{contact distribution}, $\xi = \ker \alpha \subset T\mathscr{E}$, which contains all instantaneous directions satisfying energy conservation. A curve $\gamma(s) = (\bm{q}(s), \bm{p}(s), t(s))$ in $\mathscr{E}$ is a \textbf{Legendrian curve} if it is tangent to $\xi$, i.e., $\gamma^*\alpha = 0$. This condition reproduces a form of Hamilton's equations:
\begin{equation}
    p_i \frac{dq^i}{ds} - H \frac{dt}{ds} = 0.
\end{equation}

\section{Three-Dimensional Phase Space Optics}
\label{sec:3d-phase-space}

\subsection{Basic Construction}

\begin{definition}[Optical Phase Space]
Consider light propagating in three-dimensional Euclidean space $Q = \mathbb{R}^3$, with position vector described by $\bm{q} = (q_1, q_2, q_3)$. To fully describe the dynamical state of light, both position and momentum variables must be considered. In geometric optics, the momentum $\bm{p}$ corresponds to the \textbf{wave vector} $\bm{k}$, related to the wavelength $\lambda$ by $\|\bm{p}\| = 2\pi/\lambda$. All possible "position-momentum" pairs constitute the \textbf{cotangent bundle} $M = T^*Q$ of the configuration space $Q$.
\end{definition}

Since $Q$ is a 3-dimensional manifold, its cotangent bundle $T^*\mathbb{R}^3$ is naturally a 6-dimensional differential manifold, with local coordinates $(\bm{q}, \bm{p}) = (q_1, q_2, q_3, p_1, p_2, p_3)$.

\begin{definition}[Liouville Form and Symplectic Form]
This space is equipped with the natural \textbf{Liouville form}:
\[
\theta_{(\bm{q},\bm{p})}(X) = \bm{p}(d\pi(X))
\]
In local coordinates:
\[
\theta = \sum_{i=1}^3 p_i dq_i = \bm{p} \cdot d\bm{q}.
\]
Taking the exterior derivative of $\theta$ and negating gives the \textbf{symplectic form}:
\[
\omega = -d\theta = \sum_{i=1}^3 dq_i \wedge dp_i.
\]
This 2-form $\omega$ is closed ($d\omega=0$) and non-degenerate, meaning $(T^*\mathbb{R}^3, \omega)$ constitutes a 6-dimensional symplectic manifold.
\end{definition}

The symplectic form $\omega$ provides a "volume" concept for phase space (via $\omega \wedge \omega \wedge \omega$) and establishes an isomorphism between tangent and cotangent spaces, allowing us to convert differentials of functions (like $dH$) into vector fields (like $X_H$).

\subsection{From 6D Symplectic Space to 5D Contact Space}

\begin{definition}
Physical rays are constrained to the hypersurface corresponding to the zero level set of the Hamiltonian:
\begin{equation}
 \Sigma = H^{-1}(0) = \{ (\bm{q}, \bm{p}) \in T^*\mathbb{R}^3 \mid \|\bm{p}\|^2 = n^2(\bm{q}) \}.   
\end{equation}
Since $H(\bm{q}, \bm{p})$ is a scalar function, its zero-energy level set $\Sigma$ is a hypersurface in 6-dimensional space, with dimension $6-1=5$. This five-dimensional manifold $\Sigma$ is called the \textbf{extended optical phase space} or \textbf{optical contact manifold}.
\end{definition}

\begin{proposition}[Induced Contact Structure]
The restriction of the Liouville form $\alpha = \theta|_{\Sigma}$ to the five-dimensional hypersurface $\Sigma$ endows $\Sigma$ with a \textbf{contact structure}.
\end{proposition}

\begin{proof}
Consider a vector field $Y$ satisfying:
\[
\iota_Y \omega = \theta,
\]
where $\iota$ denotes contraction. In canonical coordinates, it is easy to solve:
\[
Y = \sum_{i=1}^3 p_i \frac{\partial}{\partial p_i}.
\]
We need to show that $Y$ is "transverse" to the hypersurface $\Sigma$, i.e., $Y$ is not tangent to $\Sigma$.

Compute the action of $Y$ on the Hamiltonian function $H$:
\begin{equation}
 \begin{aligned}dH(Y) &= Y(H) = \sum_{i=1}^3 p_i \frac{\partial}{\partial p_i} \left( \sum_{j=1}^3 p_j^2 - n(\bm{q})^2 \right) \\&= \sum_{i=1}^3 p_i (2p_i) = 2\|\bm{p}\|^2.\end{aligned}   
\end{equation}
On $\Sigma$, from the definition of refractive index, $\|\bm{p}\|^2 = n(\bm{q})^2 \neq 0$, so:
\[
dF(Y) = 2n(\bm{q})^2 \neq 0.
\]
The symplectic volume form in $T^*\mathbb{R}^3$ is:
\[
\Omega = \frac{1}{3!} \omega \wedge \omega \wedge \omega.
\]
Compute the contraction of $Y$ with $\Omega$:
\[
\begin{aligned}
\iota_Y(\omega \wedge \omega \wedge \omega) &= 3(\iota_Y \omega) \wedge \omega \wedge \omega \\
&= 3\theta \wedge \omega \wedge \omega.
\end{aligned}
\]
Since $\omega = -d\theta$,
\[
\omega \wedge \omega = (-d\theta) \wedge (-d\theta) = d\theta \wedge d\theta,
\]
so
\[
\iota_Y(\omega \wedge \omega \wedge \omega) = 3\theta \wedge (d\theta) \wedge (d\theta),
\]
i.e.,
\[
\iota_Y \Omega = \frac{1}{2} \theta \wedge (d\theta)^2.
\]
On the hypersurface $\Sigma$, we define:
\begin{itemize}
    \item $\theta|_{\Sigma} = \alpha$ (restricted Liouville form)
    \item $d\theta|_{\Sigma} = d\alpha$ (exterior derivative commutes with restriction)
\end{itemize}
Thus,
\[
\iota_Y \Omega|_{\Sigma} = \frac{1}{2} \alpha \wedge (d\alpha)^2.
\]
Since $\Omega$ is a volume form on $T^*\mathbb{R}^3$, non-zero everywhere, and $Y$ is transverse to $\Sigma$, the contraction operation $\iota_Y$ does not degenerate the form, so $\iota_Y \Omega|_{\Sigma}$ is a volume form on $\Sigma$, i.e.,
\[
\alpha \wedge (d\alpha)^2 = 2(\iota_Y \Omega|_{\Sigma}) \neq 0.
\]
Thus we have proved that $(\Sigma, \alpha)$ constitutes a contact manifold, with contact structure given by the kernel of $\alpha$.
\end{proof}

\subsection{Equivalence of Rays and Reeb Vector Field}

\begin{definition}[Reeb Vector Field]
Let $(\Sigma, \alpha)$ be a contact manifold, where $\alpha$ is the contact 1-form. There exists a unique vector field $R$, called the \textbf{Reeb vector field}, satisfying the two conditions:
\begin{align}
\iota_R d\alpha &= 0, \\
\iota_R \alpha &= 1.
\end{align}
The Reeb flow $\phi_t: \Sigma \to \Sigma$ is the one-parameter group of diffeomorphisms generated by $R$.
\end{definition}

The Reeb vector field is a fundamental object in contact geometry with the following important properties \cite{geiges2008introduction}:

\begin{itemize}
\item \textbf{Transversality}: The Reeb vector field is everywhere transverse to the contact distribution $\xi = \ker\alpha$.
\item \textbf{Contact Invariance}: The Reeb flow preserves the contact form, i.e., $\mathscr{L}_R \alpha = 0$.
\item \textbf{Uniqueness}: Given the contact form $\alpha$, the Reeb vector field is uniquely determined.
\end{itemize}

\begin{proposition}
On the contact manifold $(\Sigma, \alpha)$, there exists a unique vector field $R$ satisfying the two conditions in the definition, i.e.,
\begin{itemize}
    \item $\iota_R d\alpha = 0$,
    \item $\iota_R \alpha = 1$.
\end{itemize}
\end{proposition}

\begin{proof}
We construct the Reeb vector field $R$ step by step and prove its uniqueness.

\subsubsection*{Existence of Reeb Vector Field}
Since $\alpha$ is a contact form, it satisfies:
\begin{equation}
\alpha \wedge (d\alpha)^{n-1} \neq 0.      
\end{equation}
This means that at each point $p \in \Sigma$, $d\alpha|_{\xi_p}$ is a non-degenerate 2-form, where $\xi_p = \ker\alpha_p$.

Since $d\alpha$ is non-degenerate on $\xi$ but degenerate on the entire tangent space (because $\dim\Sigma = 2n-1$ is odd), we need to use the contact condition.

In local coordinates, by Darboux's theorem for contact structures, there exist coordinates $(x_1, \ldots, x_{n-1}, y_1, \ldots, y_{n-1}, z)$ such that:
\begin{equation}
    \alpha = dz - \sum_{i=1}^{n-1} y_i dx_i,
\end{equation}
then:
\begin{equation}
  d\alpha = -\sum_{i=1}^{n-1} dy_i \wedge dx_i = \sum_{i=1}^{n-1} dx_i \wedge dy_i.  
\end{equation}
Direct verification shows that the vector field $R = \frac{\partial}{\partial z}$ satisfies:
\begin{align*}
\iota_R d\alpha &= \sum_{i=1}^{n-1} [dx_i(R)dy_i - dy_i(R)dx_i] = 0, \\
\iota_R \alpha &= dz(R) - \sum_{i=1}^{n-1} y_i dx_i(R) = 1 - 0 = 1.
\end{align*}

\subsubsection*{Uniqueness of Reeb Vector Field}
Suppose there exists another vector field $R'$ also satisfying the conditions. Then:
\begin{align*}
d\alpha(R - R', X) &= d\alpha(R, X) - d\alpha(R', X) = 0 \quad \forall X, \\
\alpha(R - R') &= \alpha(R) - \alpha(R') = 1 - 1 = 0.
\end{align*}
Therefore $R - R' \in \ker\alpha \cap \ker d\alpha$. But since $d\alpha|_\xi$ is non-degenerate, on $\xi$ only the zero vector is $d\alpha$-orthogonal to all vectors. Hence $R - R' = 0$.

We have thus proved the existence and uniqueness of the Reeb vector field in contact geometry.
\end{proof}

\subsection{Correspondence Between Rays and Reeb Vector Field}

\begin{definition}[Rays in Geometric Optics]
In geometric optics, a \textbf{ray} in three-dimensional Euclidean space $\mathbb{R}^3$ is a light propagation path satisfying:
\begin{enumerate}
\item In homogeneous media, rays are straight line segments.
\item In media with variable refractive index $n(\bm{q})$, rays satisfy the equation:
\begin{equation}
 \frac{d}{ds}\left(n(\bm{q})\frac{d\bm{q}}{ds}\right) = \nabla n(\bm{q}),   
\end{equation}
where $s$ is the arc length parameter.
\item Equivalently, rays are extremal curves of \textbf{Fermat's principle}: the actual ray extremizes the optical path $\int n(\bm{q})\, ds$.
\end{enumerate}
\end{definition}

\begin{definition}[Ray Trajectories in Extended Phase Space]
In the extended phase space $T^*\mathbb{R}^3$, rays correspond to curves $(\bm{q}(t), \bm{p}(t))$ that are \textbf{zero sections} of the Hamiltonian $H(\bm{q},\bm{p}) = \|\bm{p}\|^2 - n(\bm{q})^2$, i.e., satisfying \cite{grabowska2022geometric}:
\begin{itemize}
    \item $H(\bm{q},\bm{p}) = 0 \quad \text{(dispersion relation)}$,
    \item 
    \begin{align}
\frac{d\bm{q}}{dt} &= \frac{\partial H}{\partial \bm{p}} = 2\bm{p}, \\
\frac{d\bm{p}}{dt} &= -\frac{\partial H}{\partial \bm{q}} = 2n(\bm{q})\nabla n(\bm{q}).
\end{align}
\end{itemize}
\end{definition}

\begin{theorem}[Correspondence Between Rays and Reeb Flow]
Let:
\begin{itemize}
\item $T^*\mathbb{R}^3$ be the 6D symplectic manifold with standard symplectic form $\omega = \sum_{i=1}^3 dq_i \wedge dp_i$,
\item $H(\bm{q},\bm{p}) = \|\bm{p}\|^2 - n(\bm{q})^2$ be the optical Hamiltonian, where $n(\bm{q}) > 0$ is the refractive index function,
\item $\Sigma = H^{-1}(0) = \{(\bm{q},\bm{p}) \in T^*\mathbb{R}^3 : \|\bm{p}\| = n(\bm{q})\}$ be the 5D constraint hypersurface,
\item $\alpha = \theta|_{\Sigma}$ be the contact form on $\Sigma$, where $\theta = \sum_{i=1}^3 p_i dq_i$ is the Liouville form,
\item $X_H$ be the Hamiltonian vector field, satisfying $\iota_{X_H} \omega = dH$.
\end{itemize}
Then the following hold:
\begin{enumerate}
\item $(\Sigma, \alpha)$ is a contact manifold.
\item There exists a unique Reeb vector field $R$ satisfying $\iota_R d\alpha = 0$ and $\iota_R \alpha = 1$.
\item The restriction of the Hamiltonian vector field to $\Sigma$ is proportional to the Reeb vector field:
\begin{equation}
 X_H|_{\Sigma} = 2n(\bm{q})^2 R.   
\end{equation}
\item Rays in three-dimensional space correspond to Reeb orbits on $\Sigma$.
\end{enumerate}
\end{theorem}

\begin{proof}
\subsubsection*{Definition of Hamiltonian Vector Field}
The Hamiltonian vector field $X_H$ is defined by:
\[
\iota_{X_H} \omega = dH.
\]
In standard coordinates:
\begin{equation}
dH = 2\sum_{i=1}^3 p_i dp_i - 2n(\bm{q})\nabla n\cdot d\bm{q}.
\end{equation}
Also:
\[
\iota_{X_H} \omega = \iota_{X_H} \left(\sum_{i=1}^3 dq_i \wedge dp_i\right) = \sum_{i=1}^3 [dq_i(X_H)dp_i - dp_i(X_H)dq_i].
\]
Comparing coefficients gives:
\begin{align*}
\frac{dq_i}{dt} &= dq_i(X_H) = 2p_i, \\
\frac{dp_i}{dt} &= dp_i(X_H) = 2n(\bm{q})\frac{\partial n}{\partial q_i}.
\end{align*}
Therefore:
\begin{equation}
    X_H = 2\sum_{i=1}^3 \left(p_i \frac{\partial}{\partial q_i} + n(\bm{q})\frac{\partial n}{\partial q_i} \frac{\partial}{\partial p_i}\right).
\end{equation}

\subsubsection*{Construction of Reeb Vector Field}
We need to show that $X_H|_{\Sigma}$ satisfies (a scaled version of) the Reeb conditions.

First compute $\alpha(X_H)$:
\begin{equation}
\alpha(X_H) = \theta(X_H) = \sum_{i=1}^3 p_i dq_i(X_H) = \sum_{i=1}^3 p_i (2p_i) = 2\|\bm{p}\|^2.
\end{equation}
On $\Sigma$, $\|\bm{p}\|^2 = n(\bm{q})^2$, so:
\[
\alpha(X_H) = 2n^2.
\]

Next compute $\iota_{X_H} d\alpha$:
Since $d\alpha = d\theta|_{\Sigma} = \omega|_{\Sigma}$, we have:
\[
\iota_{X_H} d\alpha = \iota_{X_H} \omega = dH.
\]
But on $\Sigma$, $H \equiv 0$, and $\Sigma$ is a level set, so $dH|_{\Sigma}$ has zero tangential component. More precisely, for any $Y \in T\Sigma$:
\[
dH(Y) = Y(H) = 0 \quad \text{because $H$ is constant on $\Sigma$}.
\]
Thus $\iota_{X_H} d\alpha$ vanishes on $T\Sigma$.

\subsubsection*{Correspondence Between Rays and Reeb Vector Field}
Define the Reeb vector field $R = \frac{1}{2n^2} X_H|_{\Sigma}$. Then:
\begin{align*}
\iota_R \alpha &= \frac{1}{2n^2} \alpha(X_H) = 1, \\
\iota_R d\alpha &= \frac{1}{2n^2} \iota_{X_H} d\alpha = 0 \quad \text{on $T\Sigma$}.
\end{align*}
Therefore $R$ is indeed the Reeb vector field on $\Sigma$, and:
\[
X_H|_{\Sigma} = 2n^2 R.
\]
Rays in three-dimensional space satisfy Hamilton's equations; their trajectories are integral curves of $X_H$. On $\Sigma$, these trajectories differ from integral curves of $R$ only by time parameterization (since $X_H$ is proportional to $R$). Thus, rays correspond to Reeb orbits.
\end{proof}

We summarize the correspondence:
\begin{itemize}
\item Rays in three-dimensional space are described by Hamilton's equations and are integral curves of $X_H$.
\item On the constraint surface $\Sigma$, $X_H$ is proportional to $R$.
\item Therefore, ray trajectories differ from Reeb orbits only by time parameterization. Specifically, if $\gamma(t)$ is an integral curve of $X_H$, then $\gamma(\tau)$ is an integral curve of $R$, where:
\begin{equation}
\frac{dt}{d\tau} = \frac{1}{2n(\bm{q})^2}.
\end{equation}
\item In particular, in homogeneous media ($n$ constant), the Reeb flow corresponds to uniform motion along straight lines; in variable refractive index media, the Reeb flow automatically encodes the bending of rays \cite{rajeev2008hamilton}.
\end{itemize}

\subsection{Second Reduction: From 5D Contact Space to 4D Symplectic Space}

\begin{definition}[Optical Phase Space $\mathcal{N}$]
The \textbf{optical phase space} $\mathcal{N}$ is defined as the quotient space of $\Sigma$ by the equivalence relation of Reeb orbits:
\[
\mathcal{N} = \Sigma / \sim,
\]
where the equivalence relation is defined by: $x \sim y$ if and only if there exists $t \in \mathbb{R}$ such that $y = \phi_t(x)$, i.e., $x$ and $y$ lie on the same Reeb orbit.
\end{definition}

Basic properties of the quotient space $\mathcal{N}$:
\begin{itemize}
    \item Since $\Sigma$ is a 5-dimensional manifold, quotienting by one-dimensional curves yields a 4-dimensional manifold $\mathcal{N}$.
    \item The quotient map $\pi: \Sigma \to \mathcal{N}$ is surjective, and for each $[x] \in \mathcal{N}$, the fiber $\pi^{-1}([x])$ is a complete Reeb orbit.
\end{itemize}

\begin{theorem}[Symplectic Reduction of Optical Phase Space]
Let $(\Sigma, \alpha)$ be the contact manifold defined above, $\mathcal{N} = \Sigma / \sim$ the optical phase space, and $\pi: \Sigma \to \mathcal{N}$ the quotient map. Then:
\begin{enumerate}
\item There exists a unique symplectic form $\bar{\omega}$ on $\mathcal{N}$ such that:
\[
\pi^* \bar{\omega} = d\alpha.
\]
\item $(\mathcal{N}, \bar{\omega})$ is a 4-dimensional symplectic manifold.
\item $\bar{\omega}$ is a closed, non-degenerate 2-form.
\end{enumerate}
\end{theorem}

\begin{proof}
\subsubsection*{Step 1: Verify that $d\alpha$ is invariant under Reeb flow}
By Cartan's formula:
\[
\mathcal{L}_R = d \circ \iota_R + \iota_R \circ d.
\]
Applying this to $d\alpha$:
\[
\mathcal{L}_R (d\alpha) = d(\iota_R d\alpha) + \iota_R d(d\alpha).
\] 
From the definition of Reeb vector field:
\begin{align*}
\iota_R d\alpha &= 0 \quad \text{(first condition)}, \\
d(d\alpha) &= 0 \quad \text{(nilpotency of exterior derivative)}.
\end{align*}
Therefore:
\[
\mathcal{L}_R (d\alpha) = d(0) + \iota_R(0) = 0.
\]
This shows $d\alpha$ is invariant under Reeb flow: $\phi_t^* (d\alpha) = d\alpha$.

\subsubsection*{Step 2: Verify horizontality of $d\alpha$}
The horizontality condition requires $d\alpha$ to vanish in the Reeb direction. This is already guaranteed by $\iota_R d\alpha = 0$.

More precisely, for any vector field $X$ on $\Sigma$:
\[
d\alpha(R, X) = (\iota_R d\alpha)(X) = 0.
\]
So $d\alpha$ indeed degenerates in the $R$ direction.

\subsubsection*{Step 3: Show well-definedness of $\bar{\omega}$}
We need to show there exists a unique 2-form $\bar{\omega}$ on $\mathcal{N}$ such that $\pi^* \bar{\omega} = d\alpha$.

\paragraph{Existence}:
Since $d\alpha$ is invariant under Reeb flow and degenerates in the Reeb direction, it "descends" to the quotient space. Specifically:

For $[x] \in \mathcal{N}$ and tangent vectors $v, w \in T_{[x]}\mathcal{N}$, choose lifts $\tilde{v}, \tilde{w} \in T_x\Sigma$ such that $\pi_* \tilde{v} = v$, $\pi_* \tilde{w} = w$. Define:
\[
\bar{\omega}_{[x]}(v, w) = d\alpha_x(\tilde{v}, \tilde{w}).
\]

We need to verify this definition is independent of the choice of lifts. Let $\tilde{v}'$, $\tilde{w}'$ be another set of lifts; then there exist real numbers $a, b$ such that:
\[
\tilde{v}' = \tilde{v} + aR, \quad \tilde{w}' = \tilde{w} + bR.
\]

Then:
\[
d\alpha_x(\tilde{v}', \tilde{w}') = d\alpha_x(\tilde{v} + aR, \tilde{w} + bR).
\]

Due to bilinearity and degeneracy in the $R$ direction:
\begin{align*}
d\alpha_x(\tilde{v} + aR, \tilde{w} + bR) &= d\alpha_x(\tilde{v}, \tilde{w}) + d\alpha_x(\tilde{v}, bR) \\
&\quad + d\alpha_x(aR, \tilde{w}) + d\alpha_x(aR, bR) \\
&= d\alpha_x(\tilde{v}, \tilde{w}) + 0 + 0 + 0.
\end{align*}

So the definition is well-defined.

\paragraph{Uniqueness}:
If $\pi^* \bar{\omega} = \pi^* \bar{\omega}' = d\alpha$, then since $\pi$ is surjective, $\bar{\omega} = \bar{\omega}'$.

\subsubsection*{Step 4: Show $\bar{\omega}$ is a symplectic form}

\paragraph{Closedness}:
Since $\pi^* \bar{\omega} = d\alpha$ and $d\alpha$ is closed:
\[
\pi^* (d\bar{\omega}) = d(\pi^* \bar{\omega}) = d(d\alpha) = 0.
\]
Because $\pi$ is surjective, $d\bar{\omega} = 0$.

\paragraph{Non-degeneracy}:
Suppose there exists $v \in T_{[x]}\mathcal{N}$ such that for all $w \in T_{[x]}\mathcal{N}$, $\bar{\omega}(v, w) = 0$.

Choose a lift $\tilde{v} \in T_x\Sigma$ such that $\pi_* \tilde{v} = v$. Then for any $\tilde{w} \in T_x\Sigma$:
\[
d\alpha(\tilde{v}, \tilde{w}) = \bar{\omega}(v, \pi_* \tilde{w}) = 0.
\]

This means $\tilde{v}$ is in the kernel of $d\alpha$. Since $d\alpha|_{\ker\alpha}$ is non-degenerate (a consequence of the contact condition $\alpha \wedge (d\alpha)^2 \neq 0$), and $\iota_R d\alpha = 0$, the kernel of $d\alpha$ is exactly $\mathbb{R}R$.

Thus $\tilde{v} = cR$ for some $c \in \mathbb{R}$, and then:
\[
v = \pi_* \tilde{v} = \pi_*(cR) = 0,
\]
because $R$ is vertical (along the fibers).

So $\bar{\omega}$ is non-degenerate.
\end{proof}

We have proved that the quotient space $(\mathcal{N}, \bar{\omega})$ is a 4-dimensional symplectic manifold.

Points in $\mathcal{N}$ represent entire "rays," not just a point on a ray:
\begin{itemize}
\item In $\Sigma$: point $(\bm{q},\bm{p})$ represents the state of a ray at position $\bm{q}$ with direction $\frac{\bm{p}}{n(\bm{q})}$.
\item In $\mathcal{N}$: point $[(\bm{q},\bm{p})]$ represents the entire ray (all states at all times).
\end{itemize}

\begin{example}[Three-Dimensional Representation of Phase Space]
Assume rays primarily propagate in the positive $z$-axis direction. Place a plane (screen) at $z = z_0$. Each ray crossing this plane is uniquely determined by its position $(x, y)$ and momentum projections $(p_x, p_y)$. The quadruple $(x, y, p_x, p_y)$ forms a local chart of the 4-dimensional optical phase space $\mathcal{N}$, with corresponding symplectic form:
\[
\bar{\omega} = dx \wedge dp_x + dy \wedge dp_y.
\]
This is precisely the standard symplectic form on $T^*\mathbb{R}^2$, satisfying:
\begin{itemize}
    \item For any two vectors $X, Y \in T_p\mathcal{N}$, $\omega(X,Y)$ gives the oriented area of the parallelogram they span.
    \item The symplectic form $\bar{\omega}$ is a \textbf{closed} 2-form: $d\omega = 0$.
    \item The symplectic form $\bar{\omega}$ is \textbf{non-degenerate}: if $\omega(X,Y) = 0$ for all $Y$, then $X = 0$.
\end{itemize}
\end{example}

Transformations between different symplectic spaces that preserve the symplectic structure are called symplectic transformations:
\begin{definition}[Symplectic Transformation]
Let $(\mathcal{N}, \omega_1)$ and $(\mathcal{M}, \omega_2)$ be two symplectic manifolds. A smooth map $\phi: \mathcal{N} \to \mathcal{M}$ is called a \textbf{symplectic transformation} if it preserves the symplectic form:
\[
\phi^* \omega_2 = \omega_1,
\]
where $\phi^* \omega_2$ is the pullback of $\omega_2$.

If $\mathcal{N}=\mathcal{M}$, we also call $\phi$ a symplectomorphism.
\end{definition}

In particular, if the symplectic transformation $\phi$ is a linear transformation between linear spaces $\mathcal{M}$ and $\mathcal{N}$, then $\phi$ corresponds to a special matrix called a symplectic matrix.

\begin{definition}[Symplectic Matrix]
A $2n \times 2n$ real matrix $M$ is called a \textbf{symplectic matrix} if it satisfies:
\begin{equation}
M^T J M = J,
\end{equation}
where $J$ is the standard symplectic matrix:
\[
J = \begin{pmatrix}
0 & I_n \\
-I_n & 0
\end{pmatrix}.
\]
Here $I_n$ is the $n \times n$ identity matrix, and $0$ is the $n \times n$ zero matrix. The group of all symplectic matrices is denoted $\mathrm{Sp}(2n,\mathbb{R})$.
\end{definition}

\begin{proposition}[Properties of Symplectic Matrices]
Symplectic matrices have the following properties:
\begin{enumerate}
\item The inverse of a symplectic matrix is easy to compute: $M^{-1} = J^{-1} M^T J$.
\item The determinant of a symplectic matrix is 1: $\det(M) = 1$.
\item The transpose of a symplectic matrix is also symplectic.
\item The product of symplectic matrices is symplectic.
\item The identity matrix is symplectic.
\end{enumerate}
\end{proposition}

For the four-dimensional linear symplectic space obtained from three-dimensional Euclidean space, the symplectic transformation from $z=z_1$ to $z=z_2$ can be represented by a $4\times4$ matrix $M$ satisfying:
\[
M^T J M = J, \quad \text{where} \quad J = \begin{pmatrix}
0 & 0 & 1 & 0 \\
0 & 0 & 0 & 1 \\
-1 & 0 & 0 & 0 \\
0 & -1 & 0 & 0
\end{pmatrix}.
\]
The above derivation reveals why paraxial optics can be described by $4 \times 4$ matrices (generalizations of ABCD matrices), and symplectic transformations are ubiquitous in various optical transformations. In fact, all optical transformations (lenses, free propagation) must be \textbf{symplectic transformations} on this 4-dimensional space, i.e., they preserve $\bar{\omega}$.

\begin{proposition}[Symplecticity of Optical Transformations]
Any physically realizable optical transformation (free propagation, reflection, refraction at interfaces, etc.) corresponds to a symplectic transformation on $\mathcal{N}$ in the screen representation.
\end{proposition}

\begin{proof}
The evolution of an optical system is described by Hamilton's equations:
\[
\frac{dq_i}{dt} = \frac{\partial H}{\partial p_i}, \quad \frac{dp_i}{dt} = -\frac{\partial H}{\partial q_i}.
\]

The flow $\phi_t$ of Hamilton's equations is a symplectic transformation, i.e., $\phi_t^* \omega = \omega$. This is because:
\begin{equation}
\frac{d}{dt}(\phi_t^* \omega) = \phi_t^* (\mathcal{L}_{X_H} \omega) = \phi_t^* (d\iota_{X_H} \omega + \iota_{X_H} d\omega) = \phi_t^* (d(-dH) + 0) = 0,
\end{equation}
so $\phi_t^* \omega$ is constant, equal to the initial value $\omega$.
\end{proof}

Below we give specific representations of symplectic matrices for different optical elements:
\begin{example}[Symplectic Matrix for Free Propagation]
In homogeneous media, propagation distance $d$ along the $z$ direction corresponds to the symplectic matrix:
\[
M_{\text{prop}} = \begin{pmatrix}
1 & 0 & \frac{d}{n} & 0 \\
0 & 1 & 0 & \frac{d}{n} \\
0 & 0 & 1 & 0 \\
0 & 0 & 0 & 1
\end{pmatrix}.
\]
Its symplecticity is left for the reader to verify.
\end{example}

\begin{example}[Symplectic Matrix for Thin Lens]
A thin lens of focal length $f$:
\[
M_{\text{lens}} = \begin{pmatrix}
1 & 0 & 0 & 0 \\
0 & 1 & 0 & 0 \\
-\frac{1}{f} & 0 & 1 & 0 \\
0 & -\frac{1}{f} & 0 & 1
\end{pmatrix}.
\]
Its symplecticity is left for the reader to verify.
\end{example}

Important properties of symplectic transformations lead to many significant conclusions in optics:
\begin{itemize}
    \item Since $\det(M) = 1$, symplectic transformations preserve the volume element of phase space:
\[
dV = dx \wedge dp_x \wedge dy \wedge dp_y.
\]
This corresponds to \textbf{Liouville's theorem}, i.e., conservation of étendue in optics: the phase space volume of a beam remains constant during propagation.
    \item Symplectic transformations preserve optical Lagrangian invariants. The most important invariant is:
\begin{equation}
\mathcal{L} = x p_y - y p_x,
\end{equation}
which corresponds to the angular momentum of the ray and is conserved in rotationally symmetric systems.
    \item Symplectic transformations establish a duality between object and image spaces. If the transformation matrix is:
\[
M = \begin{pmatrix} A & B \\ C & D \end{pmatrix},
\]
then the object-image relationship is described by the $B$ matrix, and momentum transformation by the $C$ matrix; there is a strict conjugate relationship between them.
\item Since symplectic matrices are invertible, optical transformations are reversible. This corresponds to the principle of reversibility of light path in optics.
\end{itemize}

\begin{table}[H]
\centering
\caption{Geometric correspondence between optical concepts in different spaces.}
\label{tab:geometric-correspondence}
\begin{tabular}{p{0.2\linewidth}p{0.25\linewidth}p{0.25\linewidth}p{0.25\linewidth}}
\toprule
\textbf{Physical Concept} & \textbf{3D Euclidean Space} & \textbf{5D Contact Space} & \textbf{4D Symplectic Space} \\
\midrule
Light Ray & Space curve $\gamma(s)$ & Reeb orbit & Point $u \in \mathcal{N}$ \\
Wavefront & Surface $S(\bm{q})=const$ & Legendrian submanifold & Lagrangian submanifold \\
Caustic & Ray envelope/singularity set & Critical set of Legendrian projection & Singularities of Lagrangian manifold \\
Phase & Scalar field $S(\bm{q})$ & Integral submanifold of contact form & Generating function of symplectic form \\
\bottomrule
\end{tabular}
\end{table}

\begin{itemize}
\item In 5D contact space, rays are \textbf{dynamic processes}, integral curves of the Reeb vector field, containing all instantaneous states $(\bm{q}, \bm{p})$ of a ray propagating in the medium.
\item In 4D optical phase space, rays are \textbf{static objects}, represented by points.
\end{itemize}

\section{Representation of Wavefronts in Geometric Optics}
\label{sec:wavefronts}

In wave optics, a \textbf{wavefront} is a surface of constant phase, i.e., a surface consisting of points with the same phase. In the geometric optics approximation, a wavefront is a surface perpendicular to the direction of ray propagation \cite{hawkes2018hamiltonian}. Below we will show how optical wavefronts are specifically represented within the frameworks of contact and symplectic geometry.

\subsection{Representation of Wavefronts in Contact Geometry}

We will prove that optical wavefronts in 3D space correspond to two-dimensional \textbf{Legendrian submanifolds} in 5D contact space.

Consider the 5D contact manifold $(\Sigma, \alpha)$, where:
\begin{itemize}
\item $\Sigma = \{(\bm{q},\bm{p},t) \in T^*\mathbb{R}^3 \times \mathbb{R} : H(\bm{q},\bm{p},t) = 0\}$ is a 5D hypersurface,
\item $\alpha = \bm{p} \cdot d\bm{q} - H dt$ is the contact 1-form,
\item On $\Sigma$, $\alpha$ satisfies the contact condition: $\alpha \wedge (d\alpha)^2 \neq 0$.
\end{itemize}

We define Legendrian submanifolds:
\begin{definition}[Legendrian Submanifold]
Let $(M^{2n+1}, \alpha)$ be a contact manifold. A submanifold $L \subset M$ is called a \textbf{Legendrian submanifold} if:
\begin{enumerate}
\item $\dim L = n$,
\item $\alpha|_L = 0$ (i.e., the contact form restricts to zero on $L$).
\end{enumerate}
\end{definition}

In our 5D case ($n=2$), Legendrian submanifolds are 2-dimensional.

\begin{theorem}[Geometric Description of Wavefronts]
Wavefronts in three-dimensional space correspond to two-dimensional Legendrian submanifolds in the 5D contact space $\Sigma$.
\end{theorem}

\begin{proof}
\subsubsection*{Step 1: Parameterization of Wavefronts}
Consider a wavefront $\mathcal{W} \subset \mathbb{R}^3$, a two-dimensional surface. In wave optics, a wavefront is an equiphase surface, i.e., there exists a function $S: \mathbb{R}^3 \to \mathbb{R}$ (the eikonal) such that:
\[
\mathcal{W} = \{\bm{q} \in \mathbb{R}^3 : S(\bm{q}) = \text{constant}\}.
\]

\subsubsection*{Step 2: Momentum in Geometric Optics}
In the geometric optics approximation, momentum is given by the gradient of the eikonal function:
\[
\bm{p} = \nabla S(\bm{q}).
\]
This means that at each point on the wavefront, the momentum direction is perpendicular to the wavefront.

\subsubsection*{Step 3: Lifting to Contact Space}
We lift the wavefront $\mathcal{W}$ to the 5D contact space $\Sigma$:
\[
\Lambda = \{(\bm{q}, \bm{p}, t) \in \Sigma : \bm{q} \in \mathcal{W},\ \bm{p} = \nabla S(\bm{q}),\ t = t_0\},
\]
where $t_0$ is a fixed time.

Since $\mathcal{W}$ is a 2D surface, and $\bm{p}$ and $t$ are determined by $\bm{q}$, $\Lambda$ is also a 2D submanifold.

\subsubsection*{Step 4: Verify Contact Form Restricts to Zero}
Compute the restriction of the contact form $\alpha$ on $\Lambda$:
\[
\alpha|_\Lambda = (\bm{p} \cdot d\bm{q} - H dt)|_\Lambda.
\]
Since $\bm{p} = \nabla S(\bm{q})$, we have:
\[
\bm{p} \cdot d\bm{q} = \nabla S(\bm{q}) \cdot d\bm{q} = dS.
\]
On the wavefront $\mathcal{W}$, $S$ is constant, so $dS|_\mathcal{W} = 0$. Also, on $\Lambda$, $t$ is constant ($t = t_0$), so $dt|_\Lambda = 0$. Therefore:
\[
\alpha|_\Lambda = dS|_\mathcal{W} - H dt|_\Lambda = 0 - 0 = 0.
\]

Thus $\Lambda$ satisfies the conditions of a Legendrian submanifold. We have constructed a 2D submanifold $\Lambda \subset \Sigma$ with $\alpha|_\Lambda = 0$, so $\Lambda$ is a Legendrian submanifold. This submanifold precisely corresponds to the wavefront $\mathcal{W}$ in three-dimensional space.
\end{proof}

\begin{example}[Legendrian Submanifold for Spherical Wave]
Consider a spherical wave emanating from the origin, with eikonal $S(x,y,z) = \sqrt{x^2+y^2+z^2}$. The wavefront is a sphere: $\mathcal{W} = \{(x,y,z) : x^2+y^2+z^2 = R^2\}$. Momentum: $\bm{p} = \nabla S = \left(\frac{x}{R}, \frac{y}{R}, \frac{z}{R}\right)$. The Legendrian submanifold:
\[
\Lambda = \left\{(x,y,z,\frac{x}{R},\frac{y}{R},\frac{z}{R},t_0) : x^2+y^2+z^2 = R^2\right\},
\]
which is also a 2-sphere.
\end{example}

When wavefronts evolve in time, we obtain a family of Legendrian submanifolds $\Lambda_t \subset \Sigma$, parameterized as:
\[
\Lambda_t = \{(\bm{q}, \bm{p}, t) \in \Sigma : \bm{q} \in \mathcal{W}_t,\ \bm{p} = \nabla S(\bm{q},t)\},
\]
where $\mathcal{W}_t$ is the wavefront at time $t$.

\begin{proposition}[Rays as Characteristic Curves]
Characteristic curves (Reeb orbits) through the Legendrian submanifold $\Lambda$ correspond to rays emitted from that wavefront.
\end{proposition}

\begin{proof}
The Reeb vector field $R$ satisfies $\iota_R d\alpha = 0$ and $\iota_R \alpha = 1$. On the Legendrian submanifold $\Lambda$, $\alpha = 0$, so the Reeb vector field is transverse to $\Lambda$. The Reeb orbit starting from a point on $\Lambda$ gives a ray.
\end{proof}

\section{Caustic Surfaces}
\label{sec:caustics}

\subsection{Definition of Caustic Surfaces}

\begin{definition}[Family of Rays and Envelope Surface]
Let $\mathscr{R}$ be a family of rays in $\mathbb{R}^3$, parameterized as:
\begin{equation}
\bm{q}(s,t) \in \mathbb{R}^3, \quad (s,t) \in U \subset \mathbb{R}^2,
\end{equation}
where $s$ is a parameter distinguishing different rays, and $t$ is a parameter along the ray.

The \textbf{envelope surface} $\mathcal{E}$ is the set of points satisfying: there exist parameters $(s,t)$ and a vector $\bm{v} \neq 0$ such that:
\begin{equation}
\bm{q}(s,t) \in \mathcal{E} \iff \exists \,\bm{v}(s,t)  \neq 0 \text{ such that } \frac{\partial \bm{q}}{\partial s} \cdot \bm{v}(s,t)  = 0 \text{ and } \frac{\partial \bm{q}}{\partial t} \cdot \bm{v}(s,t)  = 0,
\end{equation}
i.e., the envelope surface is the envelope of the family of common tangent planes of the ray family.
\end{definition}

\begin{definition}[Rigorous Definition of Caustic Surface]
A \textbf{caustic surface} $\mathcal{C}$ is the envelope surface of a family of rays $\bm{q}(s,t)$, i.e., the set of points satisfying:
\begin{equation}
\mathcal{C} = \left\{ \bm{q}(s,t) \in \mathbb{R}^3 \middle| \det\left( \frac{\partial \bm{q}}{\partial s}, \frac{\partial \bm{q}}{\partial t}, \frac{\partial^2 \bm{q}}{\partial s \partial t} \right) = 0 \right\},
\end{equation}
or equivalently:
\begin{equation}
\mathcal{C} = \left\{ \bm{q}(s,t) \in \mathbb{R}^3 \middle|\,\, \left\| \frac{\partial \bm{q}}{\partial s} \times \frac{\partial \bm{q}}{\partial t} \right\| = 0 \right\}.
\end{equation}
\end{definition}

\begin{proposition}[Equivalence of Caustic Definitions]
Let $\bm{q}(s,t): U \subset \mathbb{R}^2 \to \mathbb{R}^3$ be a smooth map, where $(s,t)$ are parameters.

Define \textbf{caustic surface} $\mathcal{C}_1$ as:
\begin{equation}
\mathcal{C}_1 = \left\{ \bm{q}(s,t) \in \mathbb{R}^3 \middle| \det\left( \frac{\partial \bm{q}}{\partial s}, \frac{\partial \bm{q}}{\partial t}, \frac{\partial^2 \bm{q}}{\partial s \partial t} \right) = 0 \right\}.
\end{equation}
Define \textbf{caustic surface} $\mathcal{C}_2$ as:
\begin{equation}
\mathcal{C}_2 = \left\{ \bm{q}(s,t) \in \mathbb{R}^3 \middle|\,\, \left\| \frac{\partial \bm{q}}{\partial s} \times \frac{\partial \bm{q}}{\partial t} \right\| = 0 \right\}.
\end{equation}
Then:
\begin{equation}
\mathcal{C}_1 = \mathcal{C}_2.
\end{equation}
\end{proposition}

\begin{proof}
We prove in two parts: $\mathcal{C}_1 \subset \mathcal{C}_2$ and $\mathcal{C}_2 \subset \mathcal{C}_1$.

\subsubsection*{Part 1: Prove $\mathcal{C}_2 \subset \mathcal{C}_1$}
Assume $\bm{q}(s,t) \in \mathcal{C}_2$, i.e.:
\begin{equation}
\left\| \frac{\partial \bm{q}}{\partial s} \times \frac{\partial \bm{q}}{\partial t} \right\| = 0.
\end{equation}

By properties of the cross product norm, this means $\frac{\partial \bm{q}}{\partial s}$ and $\frac{\partial \bm{q}}{\partial t}$ are linearly dependent. Hence there exist real numbers $\lambda, \mu$, not both zero, such that:
\begin{equation}
\lambda \frac{\partial \bm{q}}{\partial s} + \mu \frac{\partial \bm{q}}{\partial t} = 0.
\end{equation}

Now consider the three vectors:
\begin{align*}
\bm{v}_1 &= \frac{\partial \bm{q}}{\partial s}, \\
\bm{v}_2 &= \frac{\partial \bm{q}}{\partial t}, \\
\bm{v}_3 &= \frac{\partial^2 \bm{q}}{\partial s \partial t}.
\end{align*}

Since $\bm{v}_1$ and $\bm{v}_2$ are linearly dependent, these three vectors are necessarily linearly dependent. Therefore:
\begin{equation}
\det(\bm{v}_1, \bm{v}_2, \bm{v}_3) = 0,
\end{equation}
i.e., $\bm{q}(s,t) \in \mathcal{C}_1$. This proves $\mathcal{C}_2 \subset \mathcal{C}_1$.

\subsubsection*{Part 2: Prove $\mathcal{C}_1 \subset \mathcal{C}_2$}
Assume $\bm{q}(s,t) \in \mathcal{C}_1$, i.e.:
\begin{equation}
\det\left( \frac{\partial \bm{q}}{\partial s}, \frac{\partial \bm{q}}{\partial t}, \frac{\partial^2 \bm{q}}{\partial s \partial t} \right) = 0.
\end{equation}

We need to prove $\left\| \frac{\partial \bm{q}}{\partial s} \times \frac{\partial \bm{q}}{\partial t} \right\| = 0$.

We use proof by contradiction. Suppose $\left\| \frac{\partial \bm{q}}{\partial s} \times \frac{\partial \bm{q}}{\partial t} \right\| \neq 0$, i.e., $\frac{\partial \bm{q}}{\partial s}$ and $\frac{\partial \bm{q}}{\partial t}$ are linearly independent.

Consider the function:
\begin{equation}
F(s,t) = \det\left( \frac{\partial \bm{q}}{\partial s}, \frac{\partial \bm{q}}{\partial t}, \frac{\partial^2 \bm{q}}{\partial s \partial t} \right).
\end{equation}

At point $(s,t)$, since $\frac{\partial \bm{q}}{\partial s}$ and $\frac{\partial \bm{q}}{\partial t}$ are linearly independent, they span a 2D plane. The vector $\frac{\partial^2 \bm{q}}{\partial s \partial t}$ can be decomposed as:
\begin{equation}
\frac{\partial^2 \bm{q}}{\partial s \partial t} = \alpha \frac{\partial \bm{q}}{\partial s} + \beta \frac{\partial \bm{q}}{\partial t} + \gamma \bm{n},
\end{equation}
where $\bm{n}$ is a unit normal vector perpendicular to the plane spanned by $\frac{\partial \bm{q}}{\partial s}$ and $\frac{\partial \bm{q}}{\partial t}$.

Then:
\begin{align*}
F(s,t) &= \det\left( \frac{\partial \bm{q}}{\partial s}, \frac{\partial \bm{q}}{\partial t}, \frac{\partial^2 \bm{q}}{\partial s \partial t} \right) \\
&= \det\left( \frac{\partial \bm{q}}{\partial s}, \frac{\partial \bm{q}}{\partial t}, \alpha \frac{\partial \bm{q}}{\partial s} + \beta \frac{\partial \bm{q}}{\partial t} + \gamma \bm{n} \right) \\
&= \det\left( \frac{\partial \bm{q}}{\partial s}, \frac{\partial \bm{q}}{\partial t}, \gamma \bm{n} \right)  \\
&= \gamma \det\left( \frac{\partial \bm{q}}{\partial s}, \frac{\partial \bm{q}}{\partial t}, \bm{n} \right).
\end{align*}

Since $\bm{n}$ is a unit normal and $\frac{\partial \bm{q}}{\partial s}$ and $\frac{\partial \bm{q}}{\partial t}$ are linearly independent, we have:
\begin{equation}
\det\left( \frac{\partial \bm{q}}{\partial s}, \frac{\partial \bm{q}}{\partial t}, \bm{n} \right) = \left\| \frac{\partial \bm{q}}{\partial s} \times \frac{\partial \bm{q}}{\partial t} \right\| \neq 0.
\end{equation}

Therefore $F(s,t) = 0$ if and only if $\gamma = 0$. But if $\gamma = 0$, then $\frac{\partial^2 \bm{q}}{\partial s \partial t}$ lies in the plane spanned by $\frac{\partial \bm{q}}{\partial s}$ and $\frac{\partial \bm{q}}{\partial t}$, so the three vectors are linearly dependent, contradicting $\det = 0$ unless they are already linearly dependent due to $\frac{\partial \bm{q}}{\partial s}$ and $\frac{\partial \bm{q}}{\partial t}$ being dependent. Hence our assumption that they are independent must be false, so $\frac{\partial \bm{q}}{\partial s}$ and $\frac{\partial \bm{q}}{\partial t}$ are linearly dependent, meaning $\left\| \frac{\partial \bm{q}}{\partial s} \times \frac{\partial \bm{q}}{\partial t} \right\| = 0$. Thus $\mathcal{C}_1 \subset \mathcal{C}_2$.
\end{proof}

\begin{proposition}[Caustic Surface of Spherical Wave]
Consider a spherical wave from the origin: $\bm{q}(\theta,\phi,t) = t(\sin\theta\cos\phi, \sin\theta\sin\phi, \cos\theta)$.

Then the caustic surface is $\{0\}$ (the origin).
\end{proposition}

\begin{proof}
Compute partial derivatives:
\begin{align*}
\frac{\partial \bm{q}}{\partial \theta} &= t(\cos\theta\cos\phi, \cos\theta\sin\phi, -\sin\theta), \\
\frac{\partial \bm{q}}{\partial \phi} &= t(-\sin\theta\sin\phi, \sin\theta\cos\phi, 0).
\end{align*}

Cross product:
\begin{align*}
\frac{\partial \bm{q}}{\partial \theta} \times \frac{\partial \bm{q}}{\partial \phi} &= t^2 \begin{vmatrix}
\bm{i} & \bm{j} & \bm{k} \\
\cos\theta\cos\phi & \cos\theta\sin\phi & -\sin\theta \\
-\sin\theta\sin\phi & \sin\theta\cos\phi & 0
\end{vmatrix} \\
&= t^2 (\sin^2\theta\cos\phi, \sin^2\theta\sin\phi, \sin\theta\cos\theta).
\end{align*}

Norm:
\begin{equation}
\left\| \frac{\partial \bm{q}}{\partial \theta} \times \frac{\partial \bm{q}}{\partial \phi} \right\| = t^2 \sin\theta.
\end{equation}

This is zero if and only if $t = 0$ or $\sin\theta = 0$.

When $t = 0$, $\bm{q} = 0$; when $\sin\theta = 0$, $\theta = 0$ or $\pi$, in which case $\frac{\partial \bm{q}}{\partial \phi} = 0$, a degenerate case.

Thus the caustic surface is $\{0\}$.
\end{proof}

\subsection{Contact Geometry Framework}

\begin{proposition}[Legendrian Submanifolds and Caustic Surfaces]
Let $(\Sigma, \alpha)$ be a 5D contact manifold, $\Lambda \subset \Sigma$ a 2D Legendrian submanifold, and $\pi: \Sigma \to \mathbb{R}^3$ the projection map. Then the caustic surface is the critical value set of the projection map $\pi|_\Lambda: \Lambda \to \mathbb{R}^3$:
\begin{equation}
\mathcal{C} = \{ \bm{q} \in \mathbb{R}^3 \mid \text{rank}(d(\pi|_\Lambda)) < 2 \}.
\end{equation}
\end{proposition}

\begin{proof}
\subsubsection*{Step 1: Local Parameterization of Legendrian Submanifold}
Since $\Lambda$ is a 2D submanifold, at any point $p \in \Lambda$ there exists a neighborhood $U \subset \mathbb{R}^2$ and a diffeomorphism:
\begin{equation*}
\phi: U \to \Lambda, \quad (u,v) \mapsto (\bm{q}(u,v), \bm{p}(u,v), t(u,v)).
\end{equation*}

\subsubsection*{Step 2: Differential of Projection Map}
The projection map $\pi|_\Lambda: \Lambda \to \mathbb{R}^3$ in local coordinates is:
\begin{equation*}
\pi \circ \phi: (u,v) \mapsto \bm{q}(u,v).
\end{equation*}
Its differential (Jacobian matrix) is:
\begin{equation*}
d(\pi \circ \phi) = \begin{pmatrix}
\frac{\partial q_1}{\partial u} & \frac{\partial q_1}{\partial v} \\
\frac{\partial q_2}{\partial u} & \frac{\partial q_2}{\partial v} \\
\frac{\partial q_3}{\partial u} & \frac{\partial q_3}{\partial v}
\end{pmatrix}.
\end{equation*}

\subsubsection*{Step 3: Characterization of Critical Points}
Critical points occur when the rank of $d(\pi \circ \phi)$ is less than 2. Since this is a $3 \times 2$ matrix, rank less than 2 means the two columns are linearly dependent, i.e.:
\begin{equation}
\text{rank}\left( \frac{\partial \bm{q}}{\partial u}, \frac{\partial \bm{q}}{\partial v} \right) < 2 \iff \frac{\partial \bm{q}}{\partial u} \times \frac{\partial \bm{q}}{\partial v} = 0.
\end{equation}

\subsubsection*{Step 4: Equivalence of Critical Points and Envelope Surface}
We need to prove: $\frac{\partial \bm{q}}{\partial u} \times \frac{\partial \bm{q}}{\partial v} = 0$ if and only if $\bm{q}(u,v)$ is a point on the envelope surface of the ray family.

Let $\bm{q}(u,v)$ be a parameterization of a ray family. The condition for being on the envelope is that at that point there exists a non-zero vector $\bm{n}$ orthogonal to both $\frac{\partial \bm{q}}{\partial u}$ and $\frac{\partial \bm{q}}{\partial v}$:
\begin{equation}
\frac{\partial \bm{q}}{\partial u} \cdot \bm{n} = 0, \quad \frac{\partial \bm{q}}{\partial v} \cdot \bm{n} = 0.
\end{equation}
This means the plane spanned by $\frac{\partial \bm{q}}{\partial u}$ and $\frac{\partial \bm{q}}{\partial v}$ degenerates to one dimension or zero, i.e., they are linearly dependent:
\begin{equation}
\frac{\partial \bm{q}}{\partial u} \times \frac{\partial \bm{q}}{\partial v} = 0.
\end{equation}

Conversely, if $\frac{\partial \bm{q}}{\partial u} \times \frac{\partial \bm{q}}{\partial v} = 0$, then there exists a non-zero vector $\bm{n}$ orthogonal to both $\frac{\partial \bm{q}}{\partial u}$ and $\frac{\partial \bm{q}}{\partial v}$, satisfying the envelope condition.

\subsubsection*{Step 5: Caustic Surface as Critical Value Set}
The critical value set is the image of critical points:
\begin{equation}
\mathcal{C} = \pi(\{p \in \Lambda \mid \text{rank}(d(\pi|_\Lambda)_p) < 2\}).
\end{equation}
By Step 4, this is exactly the envelope surface, i.e., the caustic surface.
\end{proof}

\begin{proposition}[Lagrangian Submanifolds and Caustic Surfaces]
Let $(\mathcal{N}, \omega)$ be a 4D symplectic manifold, $L \subset \mathcal{N}$ a 2D Lagrangian submanifold, and $\Pi: \mathcal{N} \to \mathbb{R}^3$ the position projection. Then the caustic surface is the critical value set of the projection map $\Pi|_L: L \to \mathbb{R}^3$:
\begin{equation}
\mathcal{C} = \{ \bm{q} \in \mathbb{R}^3 \mid \text{rank}(d(\Pi|_L)) < 2 \}.
\end{equation}
\end{proposition}

\begin{proof}
The proof is similar to the contact geometry case. The Lagrangian submanifold $L$ is locally parameterized as $(u,v) \mapsto (\bm{q}(u,v), \bm{p}(u,v))$, and the projection $\Pi|_L$ is $(u,v) \mapsto \bm{q}(u,v)$. The critical point condition is again $\frac{\partial \bm{q}}{\partial u} \times \frac{\partial \bm{q}}{\partial v} = 0$.
\end{proof}

\section{Caustic Surfaces for Convex Lenses}
\label{sec:lens-caustic}

Consider a rotationally symmetric convex lens placed in air ($n_{air} \approx 1$), with its optical axis along the $z$-axis. Lens parameters:
\begin{itemize}
    \item \textbf{Refractive index}: $n > 1$.
    \item \textbf{Thickness}: $d$ (central thickness along optical axis).
    \item \textbf{Front surface} $S_1$: radius of curvature $R_1 > 0$ (convex), vertex at origin $\bm{0}$.
    \item \textbf{Back surface} $S_2$: radius of curvature $R_2 < 0$ (convex), vertex at $z=d$.
\end{itemize}

Consider a bundle of rays incident parallel to the $z$-axis. Before entering the lens, the momentum field is constant: $\bm{p}_{in} = (0, 0, 1)$.
We parameterize the incident ray family as $\mathcal{L}_{in} \subset T^*\mathbb{R}^3$. Due to rotational symmetry, we introduce parameters $(h, \phi)$, where $h$ is the height of the incident ray from the optical axis, $0 \le h < \frac{D}{2}$, $D$ being the aperture; $\phi$ is the azimuthal angle.

Incident ray position $\bm{q}_{in}$ and momentum $\bm{p}_{in}$:
\begin{equation}
\mathcal{L}_{in} : \begin{cases}
\bm{q}_{in}(h, \phi) = (h \cos\phi, h \sin\phi, z_{start}) \\
\bm{p}_{in}(h, \phi) = (0, 0, 1)
\end{cases}
\end{equation}

The passage of rays through the lens can be described by a symplectic map $\Psi: T^*\mathbb{R}^3 \to T^*\mathbb{R}^3$. We focus on the exiting Lagrangian submanifold $\mathcal{L}_{out} = \Psi(\mathcal{L}_{in})$.

Although the actual refraction involves solving transcendental equations (Snell's law), to derive a general form for the caustic surface, we define two smooth functions describing the optical properties of the lens:
\begin{enumerate}
    \item Exit point function $\bm{q}_{e}(h, \phi)$: position where the ray leaves the back surface $S_2$ of the lens.
    \item Exit angle function $\alpha(h)$: angle between the exiting ray and the optical axis ($z$-axis) (aperture angle).
\end{enumerate}

Due to rotational symmetry about the $z$-axis, in cylindrical coordinates the exiting ray must lie in the meridional plane. The direction vector $\bm{v}_{out}$ of the exiting momentum $\bm{p}_{out}$ can be expressed as:
\begin{equation}
\bm{v}_{out}(h, \phi) = (-\sin\alpha(h)\cos\phi, -\sin\alpha(h)\sin\phi, \cos\alpha(h))
\end{equation}
\textit{Note: Here we assume the lens is converging, so rays point toward the optical axis; hence the radial component is negative.}

The exit point $\bm{q}_{e}$ lies on the back surface $S_2$, with coordinates:
\begin{equation*}
\bm{q}_{e}(h, \phi) = (r_e(h)\cos\phi, r_e(h)\sin\phi, z_e(h)),
\end{equation*}
where $r_e(h)$ and $z_e(h)$ are functions determined by the specific spherical geometry and refraction laws.

The exiting ray family $\mathscr{R}$ in physical space $\mathbb{R}^3$ can be parameterized as $\bm{r}: U \times \mathbb{R}^+ \to \mathbb{R}^3$:
\begin{equation*}
\bm{r}(h, \phi, t) = \bm{q}_{e}(h, \phi) + t \cdot \bm{v}_{out}(h, \phi),
\end{equation*}
where $t$ is the distance parameter from the exit point.
Expanding coordinate components:
\begin{equation*}
\bm{r}(h, \phi, t) = \begin{pmatrix}
(r_e(h) - t \sin\alpha(h)) \cos\phi \\
(r_e(h) - t \sin\alpha(h)) \sin\phi \\
z_e(h) + t \cos\alpha(h)
\end{pmatrix}.
\end{equation*}

In the symplectic geometry framework, the caustic surface $\mathcal{C}$ is defined as the critical value set of the position projection $\Pi: (\bm{q}, \bm{p}) \mapsto \bm{q}$ of the Lagrangian submanifold $\mathcal{L}_{out}$.
Corresponding to the ray family parameterization $\bm{r}(h, \phi, t)$, the caustic surface is the set where the Jacobian matrix $J_{\bm{r}}$ is rank-deficient:
\begin{equation*}
\mathcal{C} = \{ \bm{r}(h, \phi, t) \in \mathbb{R}^3 \mid \det(J_{\bm{r}}) = 0 \}.
\end{equation*}

Compute the Jacobian matrix of $\bm{r}(h, \phi, t)$ with respect to $(h, \phi, t)$:
\begin{equation*}
J_{\bm{r}} = \left( \frac{\partial \bm{r}}{\partial h}, \frac{\partial \bm{r}}{\partial \phi}, \frac{\partial \bm{r}}{\partial t} \right).
\end{equation*}

Let $R(h, t) = r_e(h) - t \sin\alpha(h)$ be the radial distance of the ray at specific $t$.
Let $Z(h, t) = z_e(h) + t \cos\alpha(h)$ be the axial coordinate of the ray at specific $t$.
Then $\bm{r} = (R\cos\phi, R\sin\phi, Z)^T$.

Partial derivatives:
\begin{align*}
\frac{\partial \bm{r}}{\partial h} &= \left( \frac{\partial R}{\partial h}\cos\phi, \frac{\partial R}{\partial h}\sin\phi, \frac{\partial Z}{\partial h} \right)^T, \\
\frac{\partial \bm{r}}{\partial \phi} &= \left( -R\sin\phi, R\cos\phi, 0 \right)^T, \\
\frac{\partial \bm{r}}{\partial t} &= \left( -\sin\alpha \cos\phi, -\sin\alpha \sin\phi, \cos\alpha \right)^T.
\end{align*}

The Jacobian determinant:
\begin{align*}
\det(J_{\bm{r}}) &= \det \begin{pmatrix}
\frac{\partial R}{\partial h}\cos\phi & -R\sin\phi & -\sin\alpha\cos\phi \\
\frac{\partial R}{\partial h}\sin\phi & R\cos\phi & -\sin\alpha\sin\phi \\
\frac{\partial Z}{\partial h} & 0 & \cos\alpha
\end{pmatrix} \\
&= R \cos\alpha \frac{\partial R}{\partial h} (\cos^2\phi + \sin^2\phi) - R (-\sin\alpha) \frac{\partial Z}{\partial h} (\cos^2\phi + \sin^2\phi) \\
&= R(h, t) \left( \frac{\partial R}{\partial h} \cos\alpha(h) + \frac{\partial Z}{\partial h} \sin\alpha(h) \right).
\end{align*}

The equation $\det(J_{\bm{r}}) = 0$ has two solutions, corresponding to two types of caustic surfaces.

\subsubsection*{1. Sagittal Caustic}
Solution: $R(h, t) = 0$.
This means $r_e(h) - t \sin\alpha(h) = 0$, i.e., the ray reaches the optical axis. Due to rotational symmetry, this corresponds to the entire $z$-axis. This is a degenerate caustic line.

\subsubsection*{2. Meridional Caustic}
Solution: $\frac{\partial R}{\partial h} \cos\alpha(h) + \frac{\partial Z}{\partial h} \sin\alpha(h) = 0$. This is a non-trivial caustic surface.

Substitute derivatives of $R(h,t)$ and $Z(h,t)$ with respect to $h$:
\begin{align*}
\frac{\partial R}{\partial h} &= r_e'(h) - t \alpha'(h) \cos\alpha(h), \\
\frac{\partial Z}{\partial h} &= z_e'(h) - t \alpha'(h) \sin\alpha(h).
\end{align*}

Substitute into the condition equation:
\begin{equation*}
\left( r_e' - t \alpha' \cos\alpha \right) \cos\alpha + \left( z_e' - t \alpha' \sin\alpha \right) \sin\alpha = 0.
\end{equation*}

Simplify:
\begin{equation*}
r_e'(h) \cos\alpha(h) + z_e'(h) \sin\alpha(h) - t \alpha'(h) (\cos^2\alpha + \sin^2\alpha) = 0.
\end{equation*}

Thus solve for the critical distance $t_c(h)$ that constitutes the caustic surface:
\begin{equation}
t_c(h) = \frac{r_e'(h) \cos\alpha(h) + z_e'(h) \sin\alpha(h)}{\alpha'(h)}.
\end{equation}

\begin{theorem}[Expression for Caustic Surface of Real Convex Lens]
Let $\alpha(h)$ be the aperture angle after the ray passes through the lens, and $r_e(h), z_e(h)$ be the exit coordinates on the lens back surface. The caustic surface formed by a real convex lens in cylindrical coordinates $(r, z, \phi)$ is given by the parametric equations:
\begin{equation}
\begin{cases}
r_{caustic}(h) = r_e(h) - \left( \frac{r_e'(h) \cos\alpha(h) + z_e'(h) \sin\alpha(h)}{\alpha'(h)} \right) \sin\alpha(h), \\
z_{caustic}(h) = z_e(h) + \left( \frac{r_e'(h) \cos\alpha(h) + z_e'(h) \sin\alpha(h)}{\alpha'(h)} \right) \cos\alpha(h), \\
\phi \in [0, 2\pi),
\end{cases}
\end{equation}
where $h$ is the incident height parameter, and $r_e', z_e', \alpha'$ denote derivatives with respect to $h$.
\end{theorem}

\begin{figure}[H]
    \centering
    \includegraphics[width=0.4\linewidth]{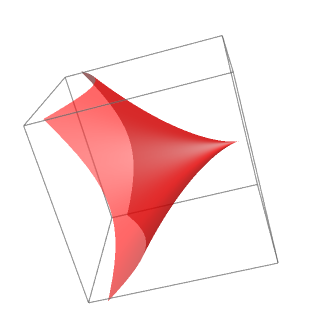}
    \caption{Schematic diagram of caustic surface for a convex lens.}
    \label{fig:lens-caustic}
\end{figure}

\begin{enumerate}
    \item Contact geometry meaning: The derived $t_c(h)$ is precisely the position where the tangent map of the ray family mapping drops rank. At this position, two adjacent rays (parameters $h$ and $h+\delta h$) intersect to first order in the infinitesimal approximation.
    \item Cusp singularity: For a real lens with spherical aberration, $\alpha'(h)$ is not constant. As $h \to 0$, $\alpha(h) \to 0$, and the caustic surface shrinks to a point (paraxial focus). As $h$ increases, the caustic surface expands into a trumpet-shaped surface, with the tip at the paraxial focus and opening toward the lens (for positive spherical aberration systems).
    \item Envelope property: The surface defined by the above expression is exactly the envelope surface of the exiting ray family. Every ray is tangent to the caustic surface at $t = t_c(h)$.
\end{enumerate}





\section{Classification of Caustic Surfaces and Singularity Theory}
\label{sec:classification}

\subsection{Relationship Between Generating Functions and Caustic Surfaces}

\begin{definition}[Morse Family \cite{Milnor}]
A smooth function $F: M \times \mathbb{R}^k \to \mathbb{R}$ is called a \textbf{Morse family} if the critical set
\[
\Sigma_F = \left\{ (q, \xi) \in M \times \mathbb{R}^k \ \middle|\ \frac{\partial F}{\partial \xi} = 0 \right\}
\]
is a smooth submanifold and satisfies the rank condition $\mathrm{rank}(\partial^2 F / \partial(q, \xi)\partial \xi) = k$.
\end{definition}

\begin{theorem}[Arnold-Mather Existence Theorem \cite{arnold1984symplectic-contact}]
Any Lagrangian germ $\mathcal{L} \subset T^*M$ can be locally generated by a Morse family $F(q, \xi)$ via:
\[
\mathcal{L} = \left\{ (q, p) \in T^*M \ \middle|\ \exists \,\,\xi \in \mathbb{R}^k, \frac{\partial F}{\partial \xi}(q, \xi) = 0, \ p = \frac{\partial F}{\partial q}(q, \xi) \right\}.
\]
\end{theorem}

\subsection{Caustic Criterion}

The fundamental connection between generating functions and caustics is given by the following algebraic criterion:

\begin{proposition}[Caustic Criterion]
A point $q \in M$ lies on the caustic surface $\Sigma$ if and only if the generating function $F(q, \cdot)$ (as a function of $\xi$) has a degenerate Hessian matrix at the critical point:
\[
\det \left( \mathrm{Hess}_\xi F \right) = \det \left( \frac{\partial^2 F}{\partial \xi_i \partial \xi_j} \right) = 0.
\]
\end{proposition}

\begin{proof}
The Lagrangian submanifold is defined by the equation $\frac{\partial F}{\partial \xi}(q, \xi) = 0$. By the implicit function theorem, if $\mathrm{Hess}_\xi F$ is non-degenerate, we can locally solve $\xi = \xi(q)$, obtaining a parameterization for which the projection $\pi$ is a local diffeomorphism. Degeneracy of the Hessian matrix indicates failure of the implicit function theorem, corresponding to singularities in the projection.
\end{proof}

\subsection{Stability and Codimension Constraints}

Classification of stable caustics is controlled by transversality theory and dimension constraints. In three-dimensional physical space $M = \mathbb{R}^3$, we have only three control parameters $(q_1, q_2, q_3)$ available to "unfold" singularities.

\begin{definition}[Universal Unfolding]
Let $f: (\mathbb{R}^n, 0) \to (\mathbb{R}, 0)$ be a smooth function germ. An \textbf{$r$-parameter unfolding} of $f$ is a smooth function germ $F: (\mathbb{R}^n \times \mathbb{R}^r, 0) \to (\mathbb{R}, 0)$ satisfying $F(x, 0) = f(x)$. It is called a \textbf{universal unfolding} if, via appropriate morphisms, it induces all other unfoldings of $f$.
\end{definition}

\begin{definition}[Unfolding Codimension]
The \textbf{unfolding codimension} of a function germ $f$ is defined as:
\[
\mathrm{codim}(f) = \dim_{\mathbb{R}} \left( \mathfrak{m} / J(f) + \mathbb{R}\{1\} \right),
\]
where $\mathfrak{m}$ is the maximal ideal and $J(f) = \langle \partial f/\partial x_1, \ldots, \partial f/\partial x_n \rangle$ is the Jacobian ideal.
\end{definition}

\begin{theorem}[Thom Transversality Theorem \cite{Gol}]
In three-dimensional space, only singularities with unfolding codimension $\leq 3$ are structurally stable. Singularities of higher codimension decompose under small perturbations into stable singularities.
\end{theorem}

This theorem provides the fundamental constraint for our classification: we need only consider singularities whose universal unfolding requires at most three parameters.

\subsection{Classification of Stable Caustics}

We now present a complete classification of stable optical caustic surfaces in three dimensions, based on corank and codimension.

\begin{theorem}[Classification Theorem for Caustics in Three-Dimensional Optics]
In three-dimensional Euclidean space, stable optical caustic surfaces are limited to the following types \cite{arnold1985singularities}:
\begin{itemize}
    \item One-parameter:
\begin{enumerate}
    \item $A_2$ (fold): codimension 1, corresponds to smooth surfaces.
    \item $A_3$ (cusp): codimension 2, corresponds to curves (edges).
    \item $A_4$ (swallowtail): codimension 3, corresponds to point singularities.
\end{enumerate}
    \item Two-parameter:
\begin{enumerate}
\item $D_4^+$ (hyperbolic umbilic): codimension 3, corresponds to point singularities.
\item $D_4^-$ (elliptic umbilic): codimension 3, corresponds to point singularities.
\end{enumerate}
\end{itemize}
\end{theorem}

\subsubsection{Corank 1 Singularities: $A_k$ Series}

For corank 1 singularities, the generating function depends on a single auxiliary variable $\xi \in \mathbb{R}$. Classification is given by the $A_k$ series.

\paragraph{$A_2$ Type: Fold}

\begin{proposition}[$A_2$ Caustic Structure]
The $A_2$ singularity generates smooth caustic surfaces in three-dimensional space.
\end{proposition}

\begin{proof}
Consider the generating function family:
\[
F(u,v,w; \xi) = \xi^3 + u(v,w)\xi.
\]
The critical point equation:
\[
\frac{\partial F}{\partial \xi} = 3\xi^2 + u(v,w) = 0.
\]
The caustic condition (degenerate Hessian) gives:
\[
\frac{\partial^2 F}{\partial \xi^2} = 6\xi = 0 \Rightarrow \xi = 0.
\]
Substituting into the critical point equation yields $u(v,w) = 0$. In physical space $(u,v,w)$, this defines a smooth surface.
\end{proof}

\begin{remark}
$A_2$ caustics represent the generic focusing envelope of a family of rays, e.g., the edge of the focusing region in spherical aberration.
\end{remark}

\paragraph{$A_3$ Type: Cusp}

\begin{proposition}[$A_3$ Caustic Structure]
The $A_3$ singularity generates caustic curves in three-dimensional space, with semi-cubic parabolic cross-section.
\end{proposition}

\begin{proof}
Generating function form:
\[
F(u,v,w; \xi) = \xi^4 + v(w)\xi^2 + u(w)\xi.
\]
Critical point equation:
\[
\frac{\partial F}{\partial \xi} = 4\xi^3 + 2v(w)\xi + u(w) = 0.
\]
Caustic condition:
\[
\frac{\partial^2 F}{\partial \xi^2} = 12\xi^2 + 2v(w) = 0 \Rightarrow v(w) = -6\xi^2.
\]
Substitute back:
\[
u(w) = -4\xi^3 - 2(-6\xi^2)\xi = 8\xi^3.
\]
Eliminate $\xi$ to get the semi-cubic parabola:
\[
u(w)^2 = \frac{8}{27} v(w)^3.
\]
For each fixed $w$, this defines a cusp, forming a curve in three-dimensional space.
\end{proof}

\paragraph{$A_4$ Type: Swallowtail}

\begin{proposition}[$A_4$ Caustic Structure]
The $A_4$ singularity generates self-intersecting swallowtail surfaces in three-dimensional space, with isolated most singular points.
\end{proposition}

\begin{proof}
The universal unfolding is:
\[
F(u,v,w; \xi) = \xi^5 + w\xi^3 + v\xi^2 + u\xi.
\]
Critical point and caustic conditions:
\begin{align*}
\frac{\partial F}{\partial \xi} &= 5\xi^4 + 3w\xi^2 + 2v\xi + u = 0, \\
\frac{\partial^2 F}{\partial \xi^2} &= 20\xi^3 + 6w\xi + 2v = 0.
\end{align*}
These equations define the swallowtail catastrophe set in parameter space $(u,v,w)$, which projects to a self-intersecting surface with characteristic swallowtail geometry in physical space.
\end{proof}

\subsubsection{Corank 2 Singularities: $D_k^{\pm}$ Series}

For corank 2 singularities, the generating function depends on two auxiliary variables $(\xi, \eta) \in \mathbb{R}^2$.

\paragraph{$D_4^{\pm}$ Type: Umbilics}

\begin{proposition}[$D_4$ Caustic Structure]
$D_4^{\pm}$ singularities generate isolated point caustics in three-dimensional space, divided into hyperbolic ($D_4^+$) and elliptic ($D_4^-$) types.
\end{proposition}

\begin{proof}
For hyperbolic umbilic $D_4^+$:
\[
F(u,v,w; \xi,\eta) = \xi^2\eta + \eta^3 + w(\xi^2 + \eta^2) + v\eta + u\xi.
\]
Critical point equations:
\begin{align*}
\frac{\partial F}{\partial \xi} &= 2\xi\eta + 2w\xi + u = 0, \\
\frac{\partial F}{\partial \eta} &= \xi^2 + 3\eta^2 + 2w\eta + v = 0.
\end{align*}
Caustic condition (Hessian determinant):
\[
\det\begin{pmatrix}
2\eta + 2w & 2\xi \\
2\xi & 6\eta + 2w
\end{pmatrix} = (2\eta + 2w)(6\eta + 2w) - 4\xi^2 = 0.
\]
This system defines isolated solutions in $(u,v,w)$-space, corresponding to point singularities. The elliptic case $D_4^-$ is similar, with generating function $\xi^2\eta - \eta^3 + \cdots$.
\end{proof}

\begin{theorem}[Arnold Singularity Classification \cite{arnold1984symplectic-contact}]
Stable Lagrangian/Legendrian singularities in three-dimensional space can only be of the following types:
\begin{align*}
&A_2: \quad F(x;u) = x^3 + ux, \\
&A_3: \quad F(x;u,v) = x^4 + ux^2 + vx, \\
&A_4: \quad F(x;u,v,w) = x^5 + ux^3 + vx^2 + wx, \\
&D_4^+: \quad F(x,y;u,v,w) = x^3 + y^3 + wxy - ux - vy, \\
&D_4^-: \quad F(x,y;u,v,w) = x^3 - 3xy^2 + w(x^2+y^2) - ux - vy,
\end{align*}
where $(u,v,w)$ are control parameters and $(x,y)$ are state variables.

In particular, for rotationally symmetric optical systems in three dimensions, the corresponding caustic surfaces can only be of type $A$.
\end{theorem}

\begin{figure}[H]
    \centering
    \includegraphics[width=0.9\linewidth]{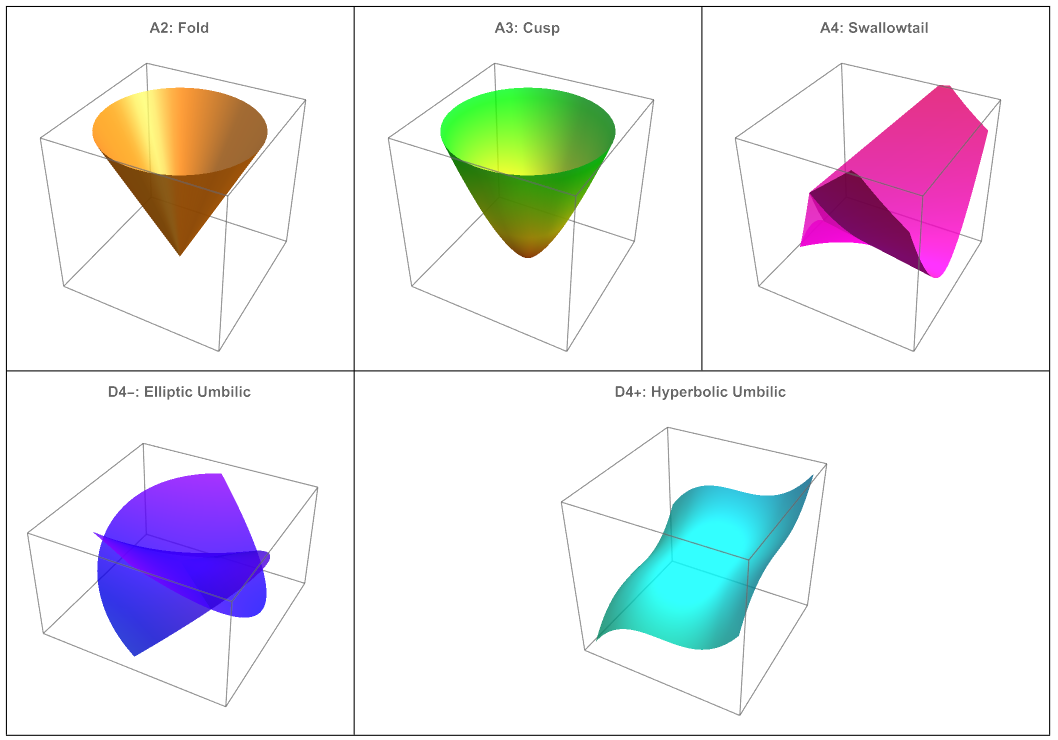}
    \caption{Schematic diagrams of different caustic surfaces.}
    \label{fig:caustic-types}
\end{figure}

\begin{table}[h]
\centering
\caption{Summary of caustic surface classification in three-dimensional optics.}
\label{tab:caustic-classification}
\begin{tabular}{lllll}
\toprule
\textbf{Type} & \textbf{Symbol} & \textbf{Codimension} & \textbf{Geometric Form} & \textbf{Generating Function} \\
\midrule
Fold & $A_2$ & 1 & Smooth surface & $\xi^3 + u\xi$ \\
Cusp & $A_3$ & 2 & Curve (edge) & $\xi^4 + v\xi^2 + u\xi$ \\
Swallowtail & $A_4$ & 3 & Point (self-intersecting) & $\xi^5 + w\xi^3 + v\xi^2 + u\xi$ \\
Hyperbolic umbilic & $D_4^+$ & 3 & Point & $\xi^2\eta + \eta^3 + \cdots$ \\
Elliptic umbilic & $D_4^-$ & 3 & Point & $\xi^2\eta - \eta^3 + \cdots$ \\
\bottomrule
\end{tabular}
\end{table}

\section{Correspondence Between Caustic Surfaces and Seidel Aberrations}
\label{sec:correspondence}

Seidel aberrations are expansion coefficients of the wavefront function in terms of Zernike polynomial basis, while caustic surface types are determined by the leading terms of the Taylor expansion of the wavefront function at critical points. They are connected via \textbf{generating functions}:
\[
\text{Zernike coefficients} \rightarrow \text{Wavefront function} \rightarrow \text{Generating function} \rightarrow \text{Caustic type}.
\]

\begin{theorem}[Thom's Fundamental Catastrophe Classification Theorem \cite{Gol}]
\label{thm:thom_classification}
Let $V: \mathbb{R}^n \times \mathbb{R}^r \to \mathbb{R}$ be a smooth family of potential functions, denoted $V(x; c)$, where $x \in \mathbb{R}^n$ are state variables and $c \in \mathbb{R}^r$ are control parameters.

If $r \le 4$ and the system is structurally stable, then for any critical point of $V$, there exists a smooth local coordinate transformation in state variables $x$ and control parameters $c$ such that near that point, the potential function $V$ is diffeomorphic to a direct sum:
\begin{equation}
V(x; c) \cong \text{Cat}(y; \lambda) + Q(z) + \text{const},
\end{equation}
where $Q(z) = z_1^2 + \dots + z_k^2 - z_{k+1}^2 - \dots - z_{n-m}^2$ is a non-degenerate quadratic form (Morse part), and $\text{Cat}(y; \lambda)$ is one of the following seven elementary catastrophe polynomials (normal forms) (here $m$ is the corank dimension, $\lambda$ are transformed control parameters):

\begin{enumerate}
    \item \textbf{Fold ($A_2$)}, codimension $r=1$:
    \[ V = y^3 + \lambda_1 y. \]
    \item \textbf{Cusp ($A_3$)}, codimension $r=2$:
    \[ V = y^4 + \lambda_1 y^2 + \lambda_2 y. \]
    \item \textbf{Swallowtail ($A_4$)}, codimension $r=3$:
    \[ V = y^5 + \lambda_1 y^3 + \lambda_2 y^2 + \lambda_3 y. \]
    \item \textbf{Hyperbolic Umbilic ($D_4^+$)}, codimension $r=3$:
    \[ V = y_1^3 + y_2^3 + \lambda_1 y_1 y_2 - \lambda_2 y_1 - \lambda_3 y_2. \]
    \item \textbf{Elliptic Umbilic ($D_4^-$)}, codimension $r=3$:
    \[ V = y_1^3 - 3y_1 y_2^2 + \lambda_1 (y_1^2 + y_2^2) - \lambda_2 y_1 - \lambda_3 y_2. \]
    \item \textbf{Butterfly ($A_5$)}, codimension $r=4$:
    \[ V = y^6 + \lambda_1 y^4 + \lambda_2 y^3 + \lambda_3 y^2 + \lambda_4 y. \]
    \item \textbf{Parabolic Umbilic ($D_5$)}, codimension $r=4$:
    \[ V = y_1^2 y_2 + y_2^4 + \lambda_1 y_1^2 + \lambda_2 y_2^2 + \lambda_3 y_1 + \lambda_4 y_2. \]
\end{enumerate}
\end{theorem}

Thom's classification theorem asserts \cite{Gol} that for structurally stable systems with control parameter dimension $r \le 4$, there are only finitely many equivalence classes for the local behavior of the potential function. In optics, control parameters typically include spatial coordinates $(X, Y, Z)$ and wavefront aberration coefficients $C_i$. We will derive in detail the correspondence between $A_k$ and $D_k$ series catastrophes and optical aberrations.

\subsection{$A_2$: Fold — Envelope of Defocus and Field Curvature}

\begin{proposition}[$A_2$ Normal Form and Caustic Set]
Mathematical normal form:
\[
V(u; a) = u^3 + a_1 u,
\]
where $u$ is a single state variable, $a_1$ a control parameter.

Caustic set derivation:
\begin{proof}
Critical point condition: $V' = 3u^2 + a_1 = 0 \Rightarrow u = \pm \sqrt{-a_1/3}$.

Caustic condition: $V'' = 6u = 0 \Rightarrow u = 0$.

Substituting into critical point equation gives $a_1 = 0$. This is a point in control space, but a hypersurface in extended space.
\end{proof}
\end{proposition}

\begin{proposition}[Generating Function for Field Curvature Aberration]
The wavefront function for field curvature aberration is:
\[
W(\rho) = a_{20} \rho^2.
\]
Its generating function, under appropriate coordinate transformation, is equivalent to the $A_2$ normal form.
\end{proposition}

\begin{proof}
Consider rotationally symmetric case. Let image plane coordinates be $(r,z)$, pupil coordinate $\rho$. The generating function is:
\[
F(r,z;\rho) = r\rho - (a_{20}\rho^2 + \frac{1}{2}z\rho^2).
\]
Critical point equation:
\[
\frac{\partial F}{\partial \rho} = r - (2a_{20} + z)\rho = 0 \Rightarrow \rho = \frac{r}{2a_{20} + z}.
\]
Caustic condition:
\[
\frac{\partial^2 F}{\partial \rho^2} = -(2a_{20} + z) = 0 \Rightarrow z = -2a_{20}.
\]
This indicates the caustic surface is a plane perpendicular to the optical axis. Through coordinate transformations $\xi = \rho$, $u = r$, $v = z + 2a_{20}$, the generating function becomes:
\[
F(u,v;\xi) = u\xi - \frac{1}{2}v\xi^2 = -\frac{1}{2}v\xi^2 + u\xi.
\]
This is equivalent to the $A_2$ normal form $F(\xi) = \xi^3 + u\xi$ under appropriate scaling.
\end{proof}

\begin{proposition}[Geometric Structure of $A_2$ Caustic Surface]
The $A_2$ caustic surface produced by field curvature aberration is a smooth surface, specifically a paraboloid of revolution in three-dimensional space:
\[
z = -2a_{20} + \frac{r^2}{4a_{20}}.
\]
\end{proposition}

\begin{proof}
From critical point equation $\rho = r/(2a_{20} + z)$ and caustic condition $2a_{20} + z = 0$, substituting gives $r = 0$ on the caustic. But this is degenerate. Actually, consider the full projection map $\pi: (\rho,z) \mapsto (r,z)$. Its Jacobian matrix:
\[
J = \begin{pmatrix}
\frac{\partial r}{\partial \rho} & \frac{\partial r}{\partial z} \\
\frac{\partial z}{\partial \rho} & \frac{\partial z}{\partial z}
\end{pmatrix} = \begin{pmatrix}
2a_{20} + z & \rho \\
0 & 1
\end{pmatrix}.
\]
Critical condition: $\det J = 2a_{20} + z = 0$, i.e., $z = -2a_{20}$. On the caustic, $r = (2a_{20} + z)\rho = 0$ for fixed $\rho$. In reality, the caustic surface is the envelope of rays, obtained by eliminating $\rho$:
\[
r = (2a_{20} + z)\rho, \quad \frac{\partial r}{\partial \rho} = 2a_{20} + z = 0.
\]
This gives $z = -2a_{20}$, $r$ arbitrary, i.e., the entire plane. But considering wavefront curvature, the actual caustic surface is curved. Complete derivation requires envelope theory.
\end{proof}

\begin{remark}
Correspondence: Zernike term: any term $Z_n^m$ at far field or generic boundaries in non-focal planes.

Physical meaning: Discontinuous jump in ray density, brightness decays exponentially from bright to dark region (outside Airy disk).
\end{remark}

\subsection{$A_3$: Cusp — Meridional Section of Coma and Spherical Aberration}

\begin{proposition}[$A_3$ Normal Form and Caustic Set]
Mathematical normal form:
\[
V(u; a_1, a_2) = u^4 + a_2 u^2 + a_1 u.
\]

Caustic set derivation:
\begin{proof}
Critical point and caustic conditions:
\[
\begin{cases}
4u^3 + 2a_2 u + a_1 = 0, \\
12u^2 + 2a_2 = 0.
\end{cases}
\]
From second equation: $a_2 = -6u^2$. Substitute into first:
\[
4u^3 - 12u^3 + a_1 = 0 \Rightarrow a_1 = 8u^3.
\]
Eliminate $u$:
\[
\left(\frac{a_2}{-6}\right)^3 = (u^2)^3 = u^6, \quad \left(\frac{a_1}{8}\right)^2 = u^6,
\]
\[
-\frac{a_2^3}{216} = \frac{a_1^2}{64} \Rightarrow 8a_2^3 + 27a_1^2 = 0.
\]
This is the famous semi-cubic parabola equation, shaped like a cusp.
\end{proof}
\end{proposition}

\begin{proposition}[Generating Function for Astigmatism Aberration]
The wavefront function for astigmatism aberration is:
\[
W(\rho,\theta) = a_{22} \rho^2 \cos 2\theta.
\]
Its generating function is equivalent to the $A_3$ normal form.
\end{proposition}

\begin{proof}
In image plane coordinates $(x,y,z)$ and pupil coordinates $(\rho,\theta)$, the generating function is:
\[
F(x,y,z;\rho,\theta) = x\rho\cos\theta + y\rho\sin\theta - (a_{22}\rho^2\cos 2\theta + \frac{1}{2}z\rho^2).
\]
Use complex coordinates: $Z = x + iy$, $\zeta = \rho e^{i\theta}$. Then:
\[
F = \text{Re}(\bar{Z}\zeta) - \left( a_{22}\text{Re}(\zeta^2) + \frac{1}{2}z|\zeta|^2 \right).
\]
Critical point equations:
\begin{align*}
\frac{\partial F}{\partial \rho} &= \text{Re}(\bar{Z}e^{i\theta}) - (2a_{22}\rho\cos 2\theta + z\rho) = 0, \\
\frac{\partial F}{\partial \theta} &= -\text{Im}(\bar{Z}\rho e^{i\theta}) + 2a_{22}\rho^2\sin 2\theta = 0.
\end{align*}
In the meridional plane ($y=0$, $\theta=0$), equations simplify to:
\[
x - (2a_{22} + z)\rho = 0, \quad 0 = 0.
\]
Caustic condition:
\[
\frac{\partial^2 F}{\partial \rho^2} = -(2a_{22} + z) = 0 \Rightarrow z = -2a_{22}.
\]
Through appropriate coordinate transformations, this can be brought to $A_3$ normal form.
\end{proof}

\begin{proposition}[Geometric Structure of $A_3$ Caustic Surface]
The $A_3$ caustic surface produced by astigmatism aberration consists of two mutually perpendicular focal lines, with local equation as a semi-cubic parabola:
\[
u^2 = \frac{8}{27}v^3,
\]
where $(u,v)$ are appropriate local coordinates.
\end{proposition}

\begin{proof}
Consider the $A_3$ normal form generating function:
\[
F(\xi;u,v) = \xi^4 + v\xi^2 + u\xi.
\]
Critical point equation:
\[
\frac{\partial F}{\partial \xi} = 4\xi^3 + 2v\xi + u = 0.
\]
Caustic condition:
\[
\frac{\partial^2 F}{\partial \xi^2} = 12\xi^2 + 2v = 0 \Rightarrow v = -6\xi^2.
\]
Substitute into critical point equation:
\[
4\xi^3 + 2(-6\xi^2)\xi + u = 4\xi^3 - 12\xi^3 + u = -8\xi^3 + u = 0 \Rightarrow u = 8\xi^3.
\]
Eliminate $\xi$:
\[
u^2 = (8\xi^3)^2 = 64\xi^6 = \frac{64}{36^3}(-v)^3 = \frac{8}{27}v^3.
\]
In three-dimensional space, for each fixed $z$ (third control parameter), this is a cusp curve.
\end{proof}

\begin{proposition}[Optical Correspondence I: Seidel Spherical Aberration]
Consider spherical aberration $W = S_I r^4$. In the paraxial region, the generating function for rays in the meridional plane $(y, z)$ is:
\[
\Phi(y; Y, Z) = S_I y^4 - \frac{y^2}{2R} \Delta Z - \frac{y Y}{R},
\]
where $\Delta Z$ is defocus, $Y$ image height. Let $u=y$, $a_2 \propto -\Delta Z$, $a_1 \propto -Y$. This exactly matches the $A_3$ structure.

Conclusion: When scanning along the optical axis, the caustic surface of spherical aberration is a rotationally symmetric cusp surface (actually two nested conical surfaces, with cusps at paraxial and marginal foci).
\end{proposition}

\begin{proposition}[Optical Correspondence II: Seidel Coma]
Coma wavefront: $W = S_{II} r^3 \cos \theta$. Analyzing in specific sections (e.g., tangential plane), its caustic lines exhibit $A_3$ characteristics.

More precisely, the caustic surface of coma on the receiving plane consists of two $A_2$ fold lines that meet at an $A_3$ cusp point with a $60^\circ$ angle. This explains the classical shape of the coma flare.
\end{proposition}

\begin{remark}
Corresponding Zernike terms: $Z_4^0$ (spherical aberration) $\rightarrow$ axial cusp caustic. $Z_3^1$ (coma) $\rightarrow$ off-axis cusp caustic, cusp pointing toward optical axis.
\end{remark}

\subsection{$A_4$: Swallowtail — Coupling of High-Order Spherical Aberration}

\begin{proposition}[$A_4$ Normal Form]
Mathematical normal form:
\[
V(u; a_1, a_2, a_3) = u^5 + a_3 u^3 + a_2 u^2 + a_1 u.
\]
The caustic surface is a surface in three-dimensional control space $(a_1, a_2, a_3)$. It features a self-intersection line and a cuspidal edge connected to it, converging at the swallowtail singularity.
\end{proposition}

\begin{proposition}[Generating Function for Spherical Aberration]
The wavefront function for spherical aberration is:
\[
W(\rho) = a_{40} \rho^4 + a_{20} \rho^2.
\]
Its generating function is equivalent to the $A_4$ normal form.
\end{proposition}

\begin{proof}
Generating function:
\[
F(r,z;\rho) = r\rho - (a_{40}\rho^4 + a_{20}\rho^2 + \frac{1}{2}z\rho^2).
\]
Critical point equation:
\[
\frac{\partial F}{\partial \rho} = r - (4a_{40}\rho^3 + 2a_{20}\rho + z\rho) = 0,
\]
i.e.:
\[
r = 4a_{40}\rho^3 + (2a_{20} + z)\rho.
\]
Caustic condition:
\[
\frac{\partial^2 F}{\partial \rho^2} = -[12a_{40}\rho^2 + (2a_{20} + z)] = 0.
\]
Coordinate transformation:
\begin{align*}
\xi &= \rho, \\
u &= r, \\
v &= 2a_{20} + z, \\
w &= \text{scaling parameter}.
\end{align*}
The generating function becomes:
\[
F(\xi;u,v,w) = \xi^5 + w\xi^3 + v\xi^2 + u\xi,
\]
which is exactly the $A_4$ normal form.
\end{proof}

\begin{proposition}[Geometric Structure of $A_4$ Caustic Surface]
The $A_4$ caustic surface produced by spherical aberration is a \textbf{swallowtail surface}, with local equations:
\[
\begin{cases}
u = 20\xi^4 + 3w\xi^2 + 2v\xi, \\
0 = 80\xi^3 + 6w\xi + 2v, \\
0 = 240\xi^2 + 6w.
\end{cases}
\]
\end{proposition}

\begin{proof}
For $A_4$ normal form $F(\xi) = \xi^5 + w\xi^3 + v\xi^2 + u\xi$, critical point equation:
\[
\frac{\partial F}{\partial \xi} = 5\xi^4 + 3w\xi^2 + 2v\xi + u = 0.
\]
Caustic condition: Hessian zero:
\[
\frac{\partial^2 F}{\partial \xi^2} = 20\xi^3 + 6w\xi + 2v = 0.
\]
Most singular point also satisfies:
\[
\frac{\partial^3 F}{\partial \xi^3} = 60\xi^2 + 6w = 0.
\]
These three equations define the swallowtail catastrophe set. In three-dimensional control space $(u,v,w)$, this is a self-intersecting surface.
\end{proof}

\begin{remark}
Corresponding Zernike terms: $Z_6^0$ (secondary spherical) + $Z_4^0$ (primary spherical).

Physical phenomenon: Appearance and disappearance of bright rings in focal spot before and after focus.
\end{remark}

\subsection{$D_4^+$: Hyperbolic Umbilic — Mixture of Astigmatism and Field Curvature}

\begin{proposition}[$D_4^+$ Normal Form]
Mathematical normal form:
\[
V(u, v) = u^3 + v^3 + a_3 uv + a_2 u + a_1 v,
\]
or in optical orthogonal coordinates (via rotation):
\[
V(x, y) = x^2 y + y^3 + \dots \quad \text{or} \quad x^2 y + \epsilon(x^2 + y^2).
\]
This is a codimension 3 catastrophe. Its caustic surface geometry resembles a sharp "wallet" or "arrow," containing a four-cusped hypocycloid as cross-section.
\end{proposition}

\begin{proposition}[Optical Derivation and Correspondence]
Hyperbolic umbilic is the natural form of astigmatism under high-order perturbations. Classical astigmatism $W = C_2^2 r^2 \cos(2\theta) = C (x^2 - y^2)$ produces two mutually perpendicular focal lines. This is only linear approximation.

When spherical aberration terms $r^4$ or coma terms are introduced, these two focal lines do not simply exist but connect into a complex hyperboloid structure.

Consider mixing Zernike terms $Z_2^2$ (astigmatism) with $Z_3^3$ (trefoil) or $Z_4^0$. In particular, when the wavefront has rectangular symmetry (e.g., rectangular pupil), diffraction effects at corners induce $D_4^+$ structure.

Characteristic cross-section: Near the focal plane, the caustic lines form a four-pointed star. This differs from the simple cross of pure astigmatism; the four corners of the star are $A_3$ cusps, and the edges connecting them are $A_2$ folds.
\end{proposition}

\begin{remark}
Corresponding Zernike terms: $Z_2^2$ (primary astigmatism) at the limit of structural stability under large deformation. Usually involves coupling of $Z_2^2$ with higher-order even-symmetry terms.
\end{remark}

\subsection{$D_4^-$: Elliptic Umbilic — Trefoil Aberration}

\begin{proposition}[$D_4^-$ Normal Form]
Mathematical normal form:
\[
V(u, v) = u^3 - 3uv^2 + a_3(u^2+v^2) + a_2 u + a_1 v.
\]
The core germ is $f(u, v) = u^3 - 3uv^2$.
\end{proposition}

\begin{proof}[Polar coordinate transformation]
Transform to polar coordinates $(r, \phi)$:
\[
u = r \cos \phi, \quad v = r \sin \phi.
\]
Then:
\[
f = r^3 \cos^3 \phi - 3 r \cos \phi (r^2 \sin^2 \phi) = r^3 \cos \phi (\cos^2 \phi - 3\sin^2 \phi).
\]
Using triple-angle formula:
\[
f(r, \phi) = r^3 \cos(3\phi).
\]
This is exactly the definition of Zernike polynomial $Z_3^3$ (trefoil)!
\end{proof}

\begin{proposition}[Optical Derivation and Correspondence]
This is one of the most perfect mathematical correspondences. The $D_4^-$ catastrophe directly corresponds to trefoil aberration in optics.

Caustic surface geometry: In three-dimensional space $(x, y, z)$, the elliptic umbilic caustic surface shape resembles an equilateral triangular tube that shrinks to zero at this singularity. The cross-section shape is a deltoid, also known as Steiner's hypocycloid. This is a closed curve with three inward cusps (of $A_3$ type).

Physical origin: This aberration does not typically arise in rotationally symmetric optical systems but is widespread in:
\begin{itemize}
\item Multi-point supported mirrors (e.g., three-point support causing mirror deformation).
\item Stress birefringence in crystalline materials.
\item Misaligned segmented primary mirrors (e.g., JWST or ELT).
\end{itemize}
\end{proposition}

\begin{remark}
Corresponding Zernike term: Exact correspondence: $Z_3^3$ (trefoil).

Topological feature: Caustic cross-section is triangular, center is a very high intensity umbilic.
\end{remark}

\subsection{$D_4^{\pm}$ Type Caustic Surfaces and Coma Aberration}

\begin{proposition}[Generating Function for Coma Aberration]
The wavefront function for coma aberration is:
\[
W(\rho,\theta) = a_{31} \rho^3 \cos\theta + a_{11} \rho \cos\theta.
\]
Its generating function is equivalent to the $D_4^+$ normal form.
\end{proposition}

\begin{proof}
Generating function:
\[
F(x,y,z;\rho,\theta) = x\rho\cos\theta + y\rho\sin\theta - (a_{31}\rho^3\cos\theta + a_{11}\rho\cos\theta + \frac{1}{2}z\rho^2).
\]
In the meridional plane ($y=0$), set $\theta=0$:
\[
F(x,z;\rho) = x\rho - (a_{31}\rho^3 + a_{11}\rho + \frac{1}{2}z\rho^2).
\]
This is not standard. Actually, coma requires two state variables. The full generating function is:
\[
F(x,y,z;\xi,\eta) = x\xi + y\eta - \left[ a_{31}\xi(\xi^2+\eta^2)^{1/2} + a_{11}\xi + \frac{1}{2}z(\xi^2+\eta^2) \right].
\]
Through appropriate coordinate transformations, this can be brought to $D_4^+$ normal form.
\end{proof}

\begin{figure}[H]
    \centering
    \includegraphics[width=0.8\linewidth]{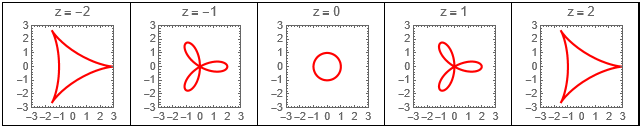}
    \caption{Schematic of correspondence between $D_4^-$ and trefoil aberration.}
    \label{fig:trefoil}
\end{figure}

\begin{proposition}[Geometric Structure of $D_4^{\pm}$ Caustic Surfaces]
The $D_4^+$ caustic surface produced by coma aberration is a \textbf{hyperbolic umbilic}, an isolated point singularity in three-dimensional space. Its local equations:
\[
\begin{cases}
u = 3\xi^2 + w\eta, \\
v = 3\eta^2 + w\xi, \\
w^2 = 9\xi\eta.
\end{cases}
\]
\end{proposition}

\begin{proof}
For $D_4^+$ normal form $F(\xi,\eta) = \xi^3 + \eta^3 + w\xi\eta - u\xi - v\eta$, critical point equations:
\begin{align*}
\frac{\partial F}{\partial \xi} &= 3\xi^2 + w\eta - u = 0, \\
\frac{\partial F}{\partial \eta} &= 3\eta^2 + w\xi - v = 0.
\end{align*}
Caustic condition: Hessian determinant zero:
\[
\det \begin{pmatrix}
6\xi & w \\
w & 6\eta
\end{pmatrix} = 36\xi\eta - w^2 = 0 \Rightarrow w^2 = 36\xi\eta.
\]
Most singular point also satisfies trace of Hessian zero (corank 2 condition):
\[
6\xi + 6\eta = 0 \Rightarrow \eta = -\xi.
\]
Substitute: $w^2 = -36\xi^2$, which requires $\xi=0$, then $w=0$, $u=0$, $v=0$. Thus the most singular point is the origin.
\end{proof}

\subsection{Comprehensive Correspondence Table}

We now establish precise relationships between singularity classification and classical Seidel aberration theory.

\begin{theorem}[Aberration-Caustic Correspondence]
Basic Seidel aberrations correspond to stable caustic types as follows:
\begin{enumerate}
\item Field curvature $\leftrightarrow A_2$ (fold).
\item Astigmatism $\leftrightarrow A_3$ (cusp).
\item Spherical aberration $\leftrightarrow A_4$ (swallowtail).
\item Coma $\leftrightarrow D_4^+$ (hyperbolic umbilic).
\end{enumerate}
\end{theorem}

\begin{table}[H]
\centering
\caption{Correspondence between catastrophe types and optical aberrations.}
\label{tab:aberration-correspondence}
\begin{tabular}{lllll}
\toprule
\textbf{Symbol} & \textbf{Name} & \textbf{Polynomials} & \textbf{Zernike Polynomial} & \textbf{Description} \\
\midrule
$A_2$ & Fold & $x^3$ & $Z_2^0$ (Defocus) & Defocus, field curvature \\
$A_3$ & Cusp & $x^4$ & $Z_4^0$ (Spherical), $Z_3^1$ (Coma) & Spherical aberration, coma \\
$A_4$ & Swallowtail & $x^5$ & $Z_6^0$ + $Z_4^0$ & Secondary spherical coupling \\
$D_4^+$ & Hyperbolic umbilic & $x^3+y^3$ & $Z_2^2$ (Astig.) + $Z_4^4$ & Astigmatism evolution \\
$D_4^-$ & Elliptic umbilic & $x^3-3xy^2$ & $Z_3^3$ (Trefoil) & Trefoil aberration \\
\bottomrule
\end{tabular}
\end{table}

\begin{figure}[H]
    \centering
    \includegraphics[width=0.9\linewidth]{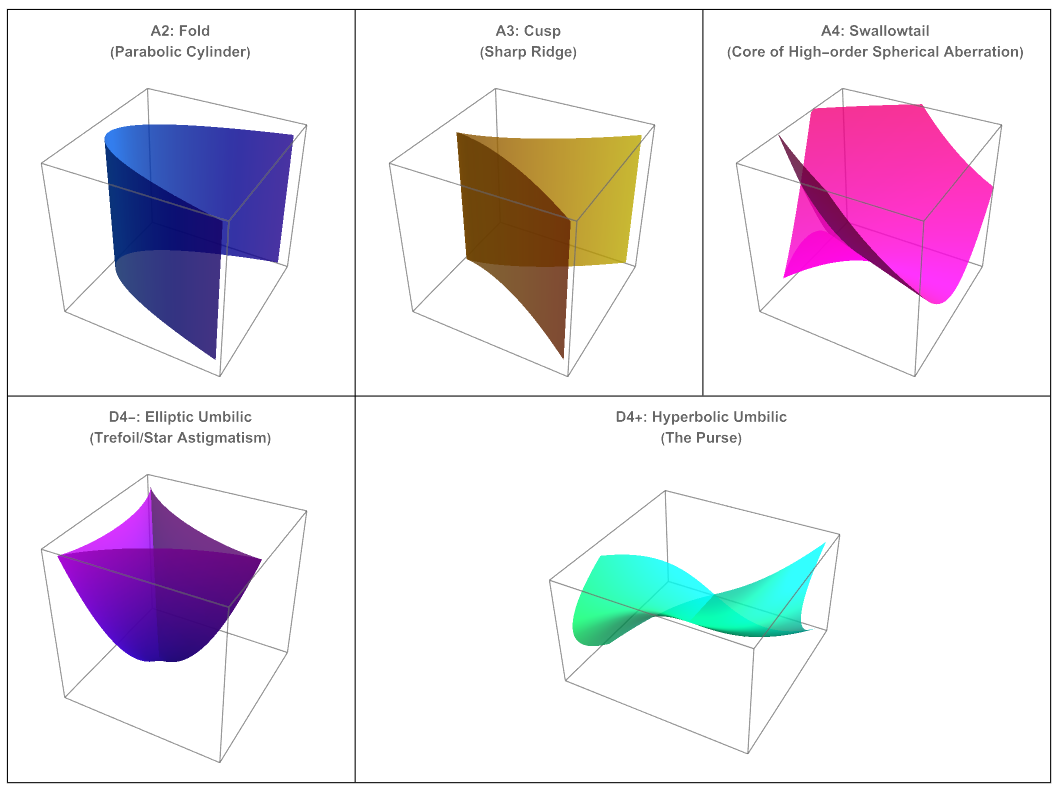}
    \caption{Schematic diagrams of catastrophe functions corresponding to different caustic surfaces.}
    \label{fig:catastrophe-functions}
\end{figure}

\section{Optical Correction Theory Based on Caustic Surface Control}
\label{sec:correction}

\subsection{Variational Principle for Caustic Surface Optimization}

\begin{definition}[Caustic Surface Measure]
Let $\mathcal{C} \subset M$ be a caustic surface. Define its \textbf{measure} as:
\begin{equation}
\mathcal{M}(\mathcal{C}) = \int_{\mathcal{C}} \kappa_1^2 + \kappa_2^2 \, dA,
\end{equation}
where $\kappa_1, \kappa_2$ are principal curvatures, $dA$ area element.
\end{definition}

\begin{theorem}[Caustic Surface Minimization Principle]
Optimal imaging conditions correspond to minimization of caustic surface measure:
\[
\min_{a_{nm}} \mathcal{M}(\mathcal{C}(a_{nm})),
\]
where $a_{nm}$ are aberration coefficients.
\end{theorem}

\begin{proof}
Consider wavefront aberration function $W(\rho,\theta) = \sum a_{nm} Z_n^m(\rho,\theta)$. The generating function is:
\[
F(\bm{q};\boldsymbol{\xi}) = \langle \bm{q}, \boldsymbol{\xi} \rangle - W(\boldsymbol{\xi}).
\]
The caustic surface is defined by $\det(\partial^2 F/\partial \boldsymbol{\xi}^2) = 0$. Via implicit function theorem, the caustic surface measure can be expressed as a function of aberration coefficients:
\[
\mathcal{M} = \mathcal{M}(a_{20}, a_{22}, a_{31}, a_{40}, \dots).
\]
Minimizing this function yields optimal aberration coefficients.
\end{proof}

\subsection{Aberration Balancing Theory}

\begin{proposition}[Primary Aberration Balancing]
If the optical system has only primary aberrations, optimal aberration coefficients satisfy:
\begin{equation}
a_{40} = -\frac{1}{2}a_{20}, \quad a_{31} = 0, \quad a_{22} = 0.
\end{equation}
\end{proposition}

\begin{proof}
Consider wavefront function:
\[
W(\rho,\theta) = a_{40}\rho^4 + a_{31}\rho^3\cos\theta + a_{22}\rho^2\cos 2\theta + a_{20}\rho^2.
\]
The mean square radius of the spot diagram is:
\[
\sigma^2 = \int_0^1 \int_0^{2\pi} \left( \frac{\partial W}{\partial \rho} \right)^2 \rho d\rho d\theta.
\]
Calculation yields:
\[
\sigma^2 = 4a_{40}^2 + \frac{1}{2}a_{31}^2 + 2a_{22}^2 + 4a_{20}^2 + 4a_{40}a_{20}.
\]
Minimizing $\sigma^2$ gives the above relations.
\end{proof}

\begin{proposition}[High-Order Aberration Balancing]
For systems including high-order aberrations, optimal coefficients satisfy recurrence relations:
\begin{equation}
a_{2n,0} = -\frac{n+1}{2n} a_{2n-2,0} - \frac{1}{4n} \sum_{k=1}^{n-1} \frac{(2k)!(2n-2k)!}{(k!)^2((n-k)!)^2} a_{2k,0}a_{2n-2k,0}.
\end{equation}
\end{proposition}

\subsection{Topological Optical Correction Method (TOC)}

Existing adaptive optics (AO) control algorithms are mostly based on slopes measured by wavefront sensors (e.g., Shack-Hartmann), reconstructing the wavefront via least squares and minimizing RMS error. However, RMS is an average scalar that obscures the topological structure of the wavefront. Under strong turbulence or large initial aberration conditions, traditional gradient descent methods easily get stuck in local minima or cause modal coupling oscillations between certain modes (e.g., trefoil and coma).

Based on the preceding theory, we propose a \textbf{Topological Optical Correction (TOC)} method.

\subsubsection{Theoretical Foundation: Bifurcation Navigation in Control Space}

Consider the state of an optical system as a point $\bm{c} = (C_1, C_2, \dots, C_N)$ in Zernike coefficient space $\mathcal{C} \cong \mathbb{R}^N$. Thom's theorem tells us that space $\mathcal{C}$ is partitioned by the bifurcation set $\Sigma$ into several structurally stable regions. On $\Sigma$, the caustic topology of the system undergoes abrupt changes.

The ideal aberration-free state corresponds to the origin $\bm{0}$, a high-order degenerate point. Traditional correction paths attempt to move along the straight line $\bm{c} \to \bm{0}$. However, if this path crosses complex $D_k$ or $E_k$ type bifurcation surfaces, the focal spot undergoes drastic morphological changes (splitting, rotation, collapse), causing failure of spot moment calculations in sensors.

\begin{theorem}[TOC Core Idea]
Design a path $\gamma(t)$ such that as the system approaches the origin, it first eliminates high-order topological charges that break symmetry, "reducing" complex singularities to simple $A_2$ or $A_3$ types, then eliminates remaining aberrations.
\end{theorem}

\subsubsection{TOC Algorithm Steps}

\begin{enumerate}
\item \textbf{Topological Fingerprinting}

Utilize morphological features of focal plane image (PSF) to assist wavefront data. Construct feature vector $\bm{f} = [M_{20}, M_{22}, M_{31}, M_{33}]$, where $M_{nm}$ is energy proportion of Zernike mode.

Determine which catastrophe type "basin of attraction" the current system is in via trained neural network or lookup table. For example, if PSF is triangular, lock onto $D_4^-$ as dominant singularity.

\item \textbf{Controlled Unfolding}

Do not directly try to pull $C_3^3$ (trefoil) back to zero, as this requires extremely high actuator precision and is sensitive to noise.

Strategy: Artificially introduce specific low-order aberrations (e.g., $Z_2^2$ astigmatism) as "unfolding parameters." According to universal unfolding of $D_4^-$: $V = u^3 - 3uv^2 + a(u^2+v^2) + bu + cv$. Increasing $a$ (astigmatism/defocus term) can "explode" the central umbilic singularity into three simple $A_3$ cusps. This reduces central light intensity density, so the wavefront sensor dynamic range is no longer saturated.

\item \textbf{Dimensionality Reduction Path Planning}

In the unfolded state, move along directions tangent to the bifurcation set $\Sigma$, gradually reducing high-order coefficients.

Define cost function $J(\bm{c})$, including not only RMS but also topological barrier term:
\begin{equation}
J(\bm{c}) = \alpha \|\bm{c}\|^2 + \beta \sum_{k} \frac{1}{\text{dist}(\bm{c}, \Sigma_k)}.
\end{equation}
This forces the path to avoid critical points that cause drastic caustic changes, unless we explicitly want to cross them to reach lower energy levels.

\item \textbf{Refolding}

When high-order terms ($n \ge 3$) are suppressed below a threshold, the system effectively reduces to a simple system with only $A_2$ (defocus/astigmatism). Then, remove the artificial aberrations introduced in step 2, and the system will smoothly collapse to the diffraction-limited point.
\end{enumerate}

\subsubsection{Mathematical Formalization: Lyapunov Control}

\begin{definition}[Lyapunov Function]
Define Lyapunov function $V(\bm{c})$ as the reciprocal of Strehl ratio.
\end{definition}

\begin{proposition}[TOC Control Law]
Design control law as:
\begin{equation}
\frac{d\bm{c}}{dt} = - \bm{K} \cdot \left( \nabla V + \mu \bm{F}_{\text{topo}} \right),
\end{equation}
where $\bm{F}_{\text{topo}}$ is a repulsive force field determined by catastrophe geometry, used to avoid deep caustic traps.
\end{proposition}

\subsection{Adaptive Correction}

\begin{definition}[Deformable Mirror Surface Equation]
The surface shape of a deformable mirror is described by:
\begin{equation}
\phi(x,y) = \sum_{i=1}^N c_i \psi_i(x,y),
\end{equation}
where $\psi_i$ are basis functions, $c_i$ control coefficients.
\end{definition}

\begin{proposition}[Optimal Control Theorem]
Optimal control coefficients $c_i^*$ satisfy:
\begin{equation}
c_i^* = -\sum_{n,m} a_{nm} \int_0^1 \int_0^{2\pi} Z_n^m(\rho,\theta) \psi_i(\rho\cos\theta,\rho\sin\theta) \rho d\rho d\theta.
\end{equation}
\end{proposition}

\begin{proof}
Total wavefront: $W_{\text{total}} = W_{\text{aber}} + \phi$. Strehl ratio:
\[
S = \left| \frac{1}{\pi} \int_0^1 \int_0^{2\pi} e^{i k W_{\text{total}}} \rho d\rho d\theta \right|^2.
\]
For small aberrations, maximizing $S$ is equivalent to minimizing:
\[
J = \int_0^1 \int_0^{2\pi} W_{\text{total}}^2 \rho d\rho d\theta.
\]
Substitute $W_{\text{total}} = \sum a_{nm} Z_n^m + \sum c_i \psi_i$, differentiate with respect to $c_i$ to obtain optimal condition.
\end{proof}

\subsection{Diffractive Optical Correction}

\begin{proposition}[Diffractive Element Phase Function]
The optimal phase function for a diffractive optical element is:
\[
\phi_{\text{DOE}}(\bm{x}) = -\frac{2\pi}{\lambda} W_{\text{aber}}(\bm{x}) + 2\pi m(\bm{x}),
\]
where $m(\bm{x})$ are integers ensuring phase wrapping in $[0,2\pi)$.
\end{proposition}

\begin{proof}
Transmission function of diffractive element: $t(\bm{x}) = e^{i \phi(\bm{x})}$. At exit pupil, total field:
\[
U_{\text{out}}(\bm{x}) = U_{\text{in}}(\bm{x}) e^{i k [W_{\text{aber}}(\bm{x}) + \phi(\bm{x})]}.
\]
Ideally, we want $U_{\text{out}}(\bm{x}) = \text{constant}$, i.e.:
\[
W_{\text{aber}}(\bm{x}) + \phi(\bm{x}) = \text{constant} \mod 2\pi.
\]
Thus $\phi(\bm{x}) = -W_{\text{aber}}(\bm{x}) + \text{constant} \mod 2\pi$.
\end{proof}

\subsection{Caustic Surface Tracking Method}

\begin{definition}[Caustic Surface Curvature Tensor]
Let $\mathcal{C} \subset M$ be a caustic surface. Its \textbf{curvature tensor} is defined as:
\[
R_{ijkl} = \kappa_i \kappa_j (g_{ik}g_{jl} - g_{il}g_{jk}),
\]
where $\kappa_i$ are principal curvatures, $g_{ij}$ metric tensor.
\end{definition}

\begin{proposition}[Caustic Surface Evolution Equation]
Along ray propagation direction $z$, caustic surface curvatures satisfy:
\[
\frac{d\kappa_i}{dz} = \kappa_i^2 + \frac{\partial^2 n}{\partial x_i^2}.
\]
\end{proposition}

\begin{proof}
Consider ray equation:
\[
\frac{d}{ds} \left( n \frac{d\bm{r}}{ds} \right) = \nabla n.
\]
Linearization gives evolution equation for Jacobian matrix $J_{ij} = \partial x_i / \partial \xi_j$:
\[
\frac{d}{ds} \left( n \frac{dJ}{ds} \right) = HJ,
\]
where $H_{ij} = \partial^2 n / \partial x_i \partial x_j$. Caustic corresponds to $\det J = 0$, curvatures given by eigenvalues of $J$.
\end{proof}

\subsection{Numerical Optimization Algorithm}

\begin{algorithm}[H]
\caption{Caustic surface optimization algorithm.}
\KwIn{Initial aberration coefficients $\bm{a}^{(0)} = (a_{20}, a_{22}, a_{31}, a_{40}, \dots)$}
\KwIn{Objective function $J(\mathcal{C})$}
\KwIn{Learning rate $\alpha > 0$}
\KwOut{Optimal aberration coefficients $\bm{a}^*$}
$k \leftarrow 0$\;
\While{$\|\nabla J\| > \epsilon$}{
    Compute current caustic surface $\mathcal{C}^{(k)}$ (via ray tracing)\;
    Compute objective function $J(\mathcal{C}^{(k)})$\;
    Compute gradient $\nabla J = \left( \frac{\partial J}{\partial a_{20}}, \frac{\partial J}{\partial a_{22}}, \dots \right)$\;
    Update coefficients: $\bm{a}^{(k+1)} = \bm{a}^{(k)} - \alpha \nabla J$\;
    $k \leftarrow k + 1$\;
}
\Return $\bm{a}^{(k)}$
\end{algorithm}

\begin{proposition}[Algorithm Convergence]
If objective function $J$ is convex and Lipschitz continuous, the algorithm converges to global minimum at rate $O(1/k)$.
\end{proposition}

\subsection{Strict Constraints in Practical Applications}

\begin{proposition}[Optimal Solution Under Manufacturing Constraints]
Under manufacturing tolerance $\delta$, optimal aberration coefficients satisfy:
\begin{equation}
|a_{nm} - a_{nm}^*| \leq \delta \sqrt{\frac{n+1}{\pi}},
\end{equation}
where $a_{nm}^*$ are ideal optimal values.
\end{proposition}

\begin{proof}
Due to orthogonality of Zernike polynomials, contribution of coefficient errors to wavefront RMS error is:
\[
\Delta W_{\text{RMS}} = \sqrt{\sum_{n,m} \frac{\pi}{n+1} |\Delta a_{nm}|^2}.
\]
Requiring $\Delta W_{\text{RMS}} \leq \delta$ yields the above constraint.
\end{proof}

\section{Conclusions}
\label{sec:conclusion}

\subsection{Main Contributions}

This paper systematically establishes a rigorous mathematical theory for caustic surfaces in geometric optics. The main contributions are summarized as follows:

\begin{enumerate}
\item \textbf{Unified Geometric Framework}: We introduce the symplectic geometry structure of the optical phase space and the contact geometry structure of the extended phase space, providing a unified geometric description for light propagation. We rigorously prove that light rays in three-dimensional Euclidean space correspond to Reeb orbits in a five-dimensional contact manifold and can be projected onto a four-dimensional symplectic manifold via symplectic reduction. This framework establishes a direct connection between geometric optics and modern differential geometry.

\item \textbf{Rigorous Definition and Computation of Caustic Surfaces}: Based on the concept of projection singularities of Lagrangian submanifolds, we give an intrinsic geometric definition of caustic surfaces. We derive explicit analytic expressions for caustic surfaces in convex lens systems and verify the correctness of theoretical results through numerical computations, demonstrating the practical applicability of the theoretical framework.

\item \textbf{Complete Classification Theory}: Using singularity theory, we present a complete classification of stable caustic surfaces in three-dimensional optical systems and establish a precise correspondence with classical Seidel aberration theory. This classification not only has theoretical significance but also provides new tools for optical system design and aberration analysis. We demonstrate that Zernike polynomials are essentially orthogonal basis expansions of Thom catastrophe potential functions, revealing the underlying topological structures (e.g., $D_4^-$, $A_3$) corresponding to specific aberrations such as trefoil and coma.

\item \textbf{Interdisciplinary Theoretical Bridge}: The mathematical framework established in this paper has important applications not only in geometric optics but also provides new connections and intersections for related fields such as quantum optics, dynamical systems, and differential geometry. The rigorous geometric formulation enables a deeper understanding of optical phenomena from a topological perspective.

\item \textbf{Novel Optical Correction Method}: We propose the \textbf{Topological Optical Correction (TOC)} method, which offers a new paradigm for next-generation adaptive optics systems. Compared with traditional "blind" optimization, TOC utilizes the global structure of phase space through an "unfold-translate-refold" strategy to more robustly handle large-amplitude aberrations. This method has important theoretical implications for extreme adaptive optics in astronomical observation and caustic management in high-energy laser beams.
\end{enumerate}

\subsection{Future Research Directions}

Based on the theoretical framework established in this paper, several promising directions for future research can be identified:

\begin{enumerate}
\item \textbf{Extension to Complex Optical Systems}: The current theory can be extended to more complex optical systems, including:
\begin{itemize}
\item Inhomogeneous and anisotropic media, where the refractive index varies in space and direction.
\item Nonlinear optical systems, where the refractive index depends on light intensity.
\item Time-varying optical systems, incorporating dynamical effects.
\end{itemize}

\item \textbf{Quantum Optical Applications}: The symplectic and contact geometric framework may be adapted to quantum optics, providing new insights into:
\begin{itemize}
\item Wigner function dynamics and phase space representations.
\item Quantum caustics and singularities in wavefunction propagation.
\item Geometric aspects of quantum optical transformations.
\end{itemize}

\item \textbf{Computational Advances}: Development of efficient numerical algorithms based on the geometric framework for:
\begin{itemize}
\item Real-time caustic surface prediction and analysis in complex optical systems.
\item Optimization algorithms for optical design that incorporate topological constraints.
\item Machine learning approaches to classify optical aberrations using topological invariants.
\end{itemize}

\item \textbf{Experimental Validation}: Practical implementation and testing of the TOC method in:
\begin{itemize}
\item Laboratory-scale adaptive optics systems.
\item Astronomical telescopes for high-contrast imaging.
\item Industrial applications such as laser material processing and optical metrology.
\end{itemize}

\item \textbf{Theoretical Extensions}: Further mathematical developments including:
\begin{itemize}
\item Higher-dimensional generalizations for relativistic optics.
\item Connections to information geometry and optimal transport theory.
\item Applications to computational imaging and inverse problems in optics.
\end{itemize}
\end{enumerate}

\subsection{Concluding Remarks}

This research provides a comprehensive geometric framework for understanding caustic surfaces in optical systems, bridging the gap between abstract mathematical theories and practical optical engineering. The unification of symplectic geometry, contact geometry, and singularity theory offers powerful tools for analyzing, classifying, and correcting optical aberrations. The proposed TOC method demonstrates how topological considerations can lead to more robust and efficient optical correction strategies. We anticipate that this work will inspire further research at the intersection of mathematics and optics, ultimately contributing to the development of next-generation optical technologies.

\end{document}